\documentclass[letterpaper, 10 pt, conference]{ieeeconf}  % Comment this line out if you need a4paper

\IEEEoverridecommandlockouts                              % This command is only needed if 
\usepackage{graphicx}    
\usepackage{epsfig} % for postscript graphics files
\usepackage{mathptmx} % assumes new font selection scheme installed
\usepackage{times} % assumes new font selection scheme installed
\usepackage{amsmath} % assumes amsmath package installed
\usepackage{amssymb}  % assumes amsmath package installed
\usepackage{graphicx} 
\usepackage{makecell} 
\usepackage{latexsym,amsfonts,amsbsy}
\usepackage{comment}
\usepackage{algorithm}
\usepackage{algorithmic}

\usepackage{xcolor}

\usepackage{flushend}    % Include this line if your 
\usepackage{comment}
\usepackage{here}
\usepackage{mathtools}
\usepackage{amsmath}
\usepackage{enumerate}
\usepackage{nccmath}
\usepackage{booktabs}
\usepackage{multirow}
\usepackage{url}

\newtheorem{proposition}{Proposition}

\newtheorem{theorem}{Theorem}
\newtheorem{corollary}{Corollary}

\newtheorem{definition}{Definition}
\newtheorem{example}{Example}
\newtheorem{remark}{Remark}

\newtheorem{problem}{Problem}

\title{\LARGE \bf
Marking Data-Informativity and Data-Driven
Supervisory Control of Discrete-Event Systems
}

\author{Yingying Liu$^{1}$, Kuma Fuchiwaki$^{1}$, and Kai Cai$^{1}$% <-this % stops a space
\thanks{$^{1}$Department of Core Informatics, Osaka Metropolitan University, 5588585 Osaka, Japan
        {\tt\small  cai@omu.ac.jp}}%
}

\begin{document}

\maketitle
\thispagestyle{empty}
\pagestyle{empty}

%\begin{frontmatter}

%\title{Fully Automated Nonblocking Heterarchical Supervisory Control of Large-Scale Discrete-Event Systems By Clustering and Abstraction\thanksref{footnoteinfo}} % Title, preferably not more 

%\thanks[footnoteinfo]{This paper was not presented at any IFAC meeting.  Corresponding author Kai Cai. %Tel. +XXXIX-VI-mmmxxi. Fax +XXXIX-VI-mmmxxv.}

%\author[Yingying]{Yingying Liu}\ead{yingyingliu611@163.com},    % Add the \author[Kai]{Zhaojian Cai},%\ead{zhaojian.cai@c.info.eng.osaka-cu.ac.jp},  % (ead) as shown\author[Kai]{Kai Cai}\ead{cai@omu.ac.jp}  % (ead) as shown\address[Yingying]{School of Mechanical and Electronic Engineering, Northwest A$\&$F University, Xianyang, China}  % Please supply              \address[Kai]{Department of Core Informatics, Osaka Metropolitan University, Osaka, Japan}        % here.
%\address[Kai]{Department of Core Informatics, Osaka Metropolitan University, Osaka, Japan}        % here.

%\begin{keyword}            % Five to ten keywords,  
%Supervisory control; Decentralized and hierarchical supervisory control; Clustering; Abstraction-based approach; Large-scale discrete-event systems.                    % chosen from the IFAC 
%\end{keyword}                             % keyword list or with the 
                                          % help of the Automatica 
                                          % keyword wizard

\begin{abstract}         
In this paper we develop a data-driven approach for marking nonblocking supervisory control of discrete-event systems (DES). We consider a setup in which models of DES to be controlled are unknown, but a set of data concerning the behaviors of DES is available. We ask the question: Under what conditions of the available data set can a valid marking noblocking supervisor be designed for the unknown DES to satisfy a given specification? Answering this question, we identify and formalize a novel concept called marking data-informativity. Moreover, we design an algorithm for the verification of this concept. Next, if the data set fails to be marking informative, we propose two related new
concepts of restricted marking data-informativity and marking informatizability. Finally, we develop an algorithm to compute the largest subset of control
specification for which the data set is least restricted marking informative.
\end{abstract}

%\end{frontmatter}

\section{Introduction}

A system in which the state transitions discretely due to the occurrence of events is called a discrete-event system (DES).
Tansportation systems, production systems, and communication systems are examples of DES \cite{Takai2007DiscreteEventSystems}, \cite{cassandras2008introduction}. Supervisory control is known as an effective control method for DES. To enforce a desired specification on a given DES, supervisory control theory aims to construct a feedback controller called a supervisor, whose control mechanism is disabling or enabling occurrences of certain (controllable) events. Autonomous driving and distribution warehouse automation are typical applications of supervisory control \cite{cai2021supervisory}, \cite{ramadge1987supervisory}, \cite{knuth1984}, \cite{wonham1987supremal}, \cite{wonham2018supervisory}, \cite{Takai2014SupervisoryControl}.

In recent years, data-driven approaches attract much attention in various fields because of the spread of IoT and rapid development of data sensing technology \cite{karimi2024determining}, \cite{schrum2024maveric}. This means the available data on many systems is growing rapidly \cite{YasukochiKeiko2023ExponentialDevelopment}, \cite{chen2024ai}. On the other hand, supervisory control of DES is mostly a model-based approach, with few studies attempting a data-driven approach \cite{konishi2022efficient}, \cite{steentjes2021data}, \cite{gu2024data}. These are the motivations for studying data-driven supervisory control.

In model-based supervisory control, a DES to be controlled is first modeled as a finite-state automaton, and its behaviors represented by regular languages, then supervisory control design is carried out based on those models. On the other hand, our data-driven approach does not aim to identify models of DES, but deals with data set concerning the behavior of DES. The goal is then to construct a valid supervisor for a family of DES models that all can generate the data set. This allows effective control over DES where the environment is unknown and modeling is difficult. In prior research, a concept of data-informativity, which captures the notion that the available data set contains sufficient information such that a valid supervisor may be constructed for a family of DES models that all can generate the data set, the necessary and sufficient conditions for data-informativity, the algorithm for its verification \cite{cai2022data}, \cite{ohtsuka2023data}. 
However, previous studies did not account for marked behaviors in their formulations. In this paper, we address this limitation by incorporating data-informativity that considers marked behaviors. This approach enhances the supervisory control design, enabling the system to avoid deadlocks while ensuring the achievement of specified goals \cite{Fuchiwaki2024Marking}.

In this thesis, we aim to initiate a systematic development of a data-driven approach for marking supervisory control of DES
and verify its effectiveness through simulation examples. Specifically, we consider a setup in which models of DES to be controlled are unknown, but three types of data concerning the behaviors of DES are available. The first type is observation data $D$, which is a collection of observed behaviors (strings of events) from the unknown DES. The second type is observation data on marked behaviors $D_m$. The third type is prior knowledge data $D^-$ about the impossible behaviors of the DES. Given these three types of data, we consider the question of what conditions of the available dataset it is possible to design a valid marking nonblocking supervisor that meets the given specifications for an unknown DES. Answering this question, we identify and formalize a novel concept called marking data-informativity. Marking data-informativity characterizes a condition that the given data set information such that a valid marking nonblocking supervisor may be constructed for a family of DES models that all can generate the data set. Thus rather than trying to first identify a model for the unknown DES, our approach based on marking data-informativity aims to directly construct from the data set a marking nonblocking supervisor valid for all possible models undistinguishable from the unknown DES. If such a marking nonblocking supervisor can be constructed, it is also valid for the unknown DES.

The structure of the thesis is as follows: 
Section~\ref{sec:preliminaries} reviews preliminaries.
Section~\ref{chap:Data-Driven Marking Supervisory Control and Marking Data-Informativity} presents data-driven marking supervisory control and marking data-informativity, while Section~\ref{chap:Restricted marking data-informativity} addresses restricted marking data-informativity.
Section~\ref{chap:Conclusions} concludes the paper.

\section{PRELIMINARIES AND PROBLEM STATEMENT}\label{sec:preliminaries}
 Let $\Sigma$ be a nonempty finite alphabet of symbols $\sigma, \alpha, \beta, \ldots$. These symbols will denote events and $\Sigma$ the event set. A string $s = \sigma_1 \sigma_2 \dots \sigma_k,k \leq 1$, is a finite sequence of events. Let $\Sigma ^ \ast $ be the set of all finite-length strings including the empty string $\epsilon$, meaning `do nothing'. A language $L$ is an arbitrary subset of strings in $\Sigma ^ \ast$, i.e. $L \subseteq \Sigma ^ \ast$. When there exists an event sequence $s_2 \in \Sigma^\ast$ such that $s = s_{1}s_2$, the event sequence $s_1\in\Sigma^\ast$ is called a prefix of $s$. For any language $L \subseteq \Sigma^\ast$, the language represented by all prefixes of its elements can be denoted as $\overline{L} := \{s_1 \in \Sigma ^ \ast | ( \exists s_2 \in \Sigma ^ \ast) s_{1}s_2\in L \}.$
It is always true that $L \subseteq \overline{L}$. We say that a language $L$ is closed if $L = \overline{L}$.

Automata serve as a fundamental model for discrete-event systems. A finite state automaton $G$ is a five-tuple
\begin{equation} \label{eq:DFA}
G := (Q,\Sigma,\delta,q_0,Q_m) 
\end{equation}
where $Q$ is the finite state set, $q_0 \in Q$ the initial state, $Q_m \subseteq Q$ the set of marker states, $\Sigma$ the finite event set, and $\delta : Q\times\Sigma \rightarrow Q$ the (partial) state transition function. 
Write $\delta(q,\sigma)!$ to mean $\sigma\in\Sigma$ is defined at state $q\in Q$, and write $\neg\delta(q,\sigma)!$ to mean $\sigma\in\Sigma$ is not defined at $q\in Q$. We say that the automaton $G$ is deterministic if
\begin{equation}
 (\forall q \in Q, \forall \sigma \in \Sigma) \hspace{0.2cm} \delta(q,\sigma)! \Longrightarrow |\delta(q,\sigma)| = 1 . \nonumber
\end{equation}
Namely, the destination state of every state transition is
unique. We shall focus exclusively on deterministic automaton
unless otherwise stated. The state transition function $\delta$
may be inductively defined, and we write $\delta(q, s)!$ to mean that $\delta(q, s)$ is defined. 
The closed behavior $L(G)$ of $G$ is $L(G):=\{s \in \Sigma^*|\delta(q_0,s)!\}.$
As defined $L(G)$ is closed. Furthermore, the marked behavior $L_m(G)$ of $G$ is the language defined as the set of all strings of $L(G)$ which $G$ can generate starting from the initial state $q_0$ to the marker state $Q_m$ is defined as
$L_m(G):=\{s \in L(\textbf{G})|\delta(q_0,s) \in Q_m\}.$
As in the case of an automaton, a state $q \in Q$ is reachable if $(\exists s \in \Sigma ^ \ast) \delta(q_0,s)! \: \& \: \delta(q_0,s)=q.$
$G$ itself is reachable if $q$ is reachable for all $q \in Q$. A state $q \in Q$ is coreachable if 
  $  (\exists s \in \Sigma ^ \ast) \delta(q,s) \in Q_m.$ 
$G$ itself is coreachable if $q$ is coreachable for all $q \in Q$.
$G$ is nonblocking if every reachable state is coreachable, or equivalently $L(G) = \overline{L_m(G)}$, namely every string in the closed behavior may be completed to a string in the marked behavior. 
Finally, $G$ is trim if it is both reachable and coreachable. Of course trim implies nonblocking, but the converse is false. 

Standard model-based supervisory control consists of three elements: the plant (to be controlled), control specification, and the supervisor (controller). The plant is modeled using a DFA $G$. The supervisor observes events occurring in the plant and forbids inappropriate events to ensure that the plant meets the control specifications.

Generally, not all strings in the marked behavior $L_m(G)$ of the plant $G$ are desired. Thus, we represent a desired behavior to be enforced on $G$ as a control specification $K \subseteq L_m(G)$. Note that a more general control specification $E \subseteq \Sigma^\ast$ may be considered. In this case, set $K := L_m(G)\cap E$ and we again have a specification $K \subseteq L_m(G)$ to be enforced on $G$.

For a mechanism to enforce a specification $K \subseteq L_m(G)$ on the plant $G$, we assume that a subset of events $\Sigma_c \subseteq \Sigma$, called the controllable events, are capable of being enabled or disabled by an external controller. On the contrary, $\Sigma_u := \Sigma \setminus \Sigma_c$ is the set of uncontrollable events, which cannot be externally disabled and must be considered permanently enabled.

Under the above mechanism, define a supervisor to be a function $V: L(G) \to Pwr (\Sigma_c)$. Here, $Pwr(\Sigma_c)$ denotes the set of all subsets of controllable events. Thus, a supervisor assigns to each string $s \in L(G)$ generated by the plant $G$ a subset of controllable events $V(s) \subseteq \Sigma_c$ to be disabled. Write $V / G$ for the closed-loop system. The language of the closed-loop system $L(V / G)$ is defined as follows:

 \begin{enumerate}
     \item $\epsilon \in L(V / G)$;
     \item if $s \in L(V / G)$ and $\sigma \in \Sigma \setminus V(s)$ and $s\sigma \in L(G)$, then $s\sigma \in L(V / G)$;
     \item no other strings belong to $L(V / G)$.
 \end{enumerate}

\noindent
Thus, $L(V / G)$ contains those strings in $L(G)$ that are not disabled by the supervisor $V$. By definition, $L(V / G) \subseteq L(G)$ and $\overline{L(V / G)} = L(V / G)$. 

In addition, the marked language of the closed-loop system $L_m(V/G)$ is defined as follows:
\begin{align*}
    L_m(V/G) = L(V/G)\cap L_m(G).
\end{align*}
The supervisor $V$ is called a {\em marking supervisor} if 
the marked language of the closed-loop system is
\begin{align*}
    L_m(V/G) = L(V/G)\cap K.
\end{align*}
Namely the supervisor $V$ takes care of the marking information of the specification $K$. Moreover $V$ is called a marking nonblocking supervisor if $\overline{L_m(V/G)} = L(V/G)$.

\begin{definition}\label{def:controllability}
(Controllability): Given a plant $G$, a control specification $K \subseteq L_m(G)$ is said to be controllable with respect to $G$ provided
\begin{align}
 (\forall s \in \overline{K}, \forall \sigma \in \Sigma_u) \hspace{0.1cm}s\sigma \in L(G) \Rightarrow s\sigma \in \overline{K}.\label{eq:controllability}
\end{align}
\end{definition}

In words, a specification $K$ is controllable with respect to $G$ if and only if any string in the prefix closure $\overline{K}$ cannot exit $\overline{K}$ on a continuation by an uncontrollable event. Namely, the prefix closure of $K$ is invariant under uncontrollable flows. Equivalently, we can write $\overline{K} \Sigma_u \cap L(G) \subseteq \overline{K}$. It is known that the specification language $K (\neq \emptyset)$ being controllable is necessary and sufficient for the existence of a marking nonblocking supervisor $V$ such that $L_m(V / G) = K$ and $\overline{L_m(V/G)}=L(V/G)$.

Suppose that $K \subseteq L_m(G)$ is controllable. Then the marking nonblocking supervisor $V: L(G) \to Pwr(\Sigma_c)$ such that $L_m(V / G) = K$ is constructed as follows:
\begin{align}
V(s) = \left\{
\begin{array}{ll}
 \{\sigma \in \Sigma_c \mid s\sigma \notin K\} &\text{if } s \in \overline{K}, \\
 \emptyset & \text{if } s \in L(G) \setminus \overline{K}. 
 \label{eq:model-based supervisor}
\end{array}
\right.
\end{align}

Whether or not $K$ is controllable, we can write $C(K)$ for the family of all controllable sublanguages of $K$:
\begin{align}
C(K) := \{ K' \subseteq K \mid \overline{K'} \Sigma_u \cap L(G) \subseteq \overline{K'}\}.
\label{eq:model-base C(K)}
\end{align}

It is known that the union of controllable sublanguages of $K$ is still a controllable sublanguage of $K$. This means that $C(K)$ is closed under set union, so $C(K)$ contains a unique supremal element:
\begin{align}
\text{sup} \, C(K) := \bigcup \{ K' \mid K' \in C(K) \}.
\label{eq:model-base supC(K)}
\end{align}

Since $\text{sup} \, C(K)$ is controllable, as long as $\text{sup} \, C(K) \neq \emptyset$, there exists a marking nonblocking supervisor $V_{\text{sup}}$ such that $L_m(V_{\text{sup}} / G) = \text{sup} \, C(K)$. In this sense, $V_{\text{sup}}$ is optimal (maximally permissive), allowing the generation by $G$ of the largest possible set of strings that satisfies a given specification.

\begin{figure}[b]
  \centering
  \includegraphics[width=80mm]{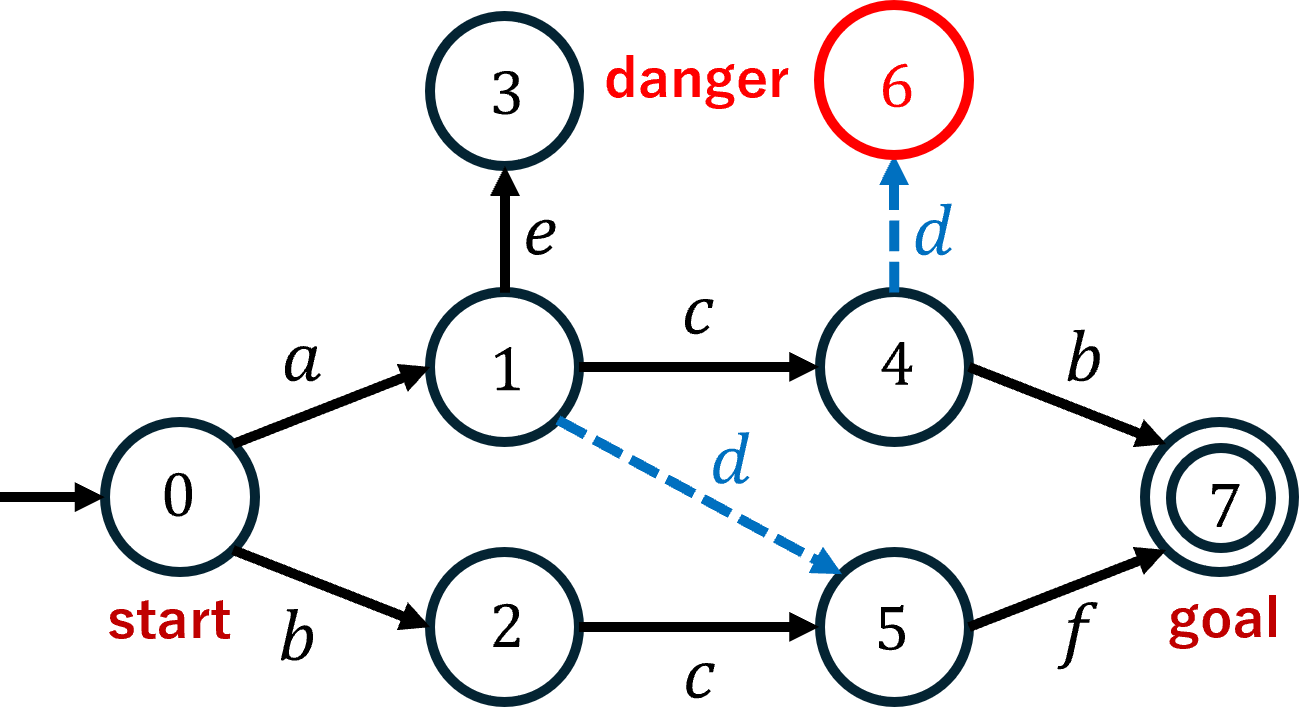}
  \caption{Robot navigation: plant $G_1$.  \label{fig:real_plant}}
\end{figure}

\begin{figure}[b]
  \centering
  \includegraphics[width=85mm]
  {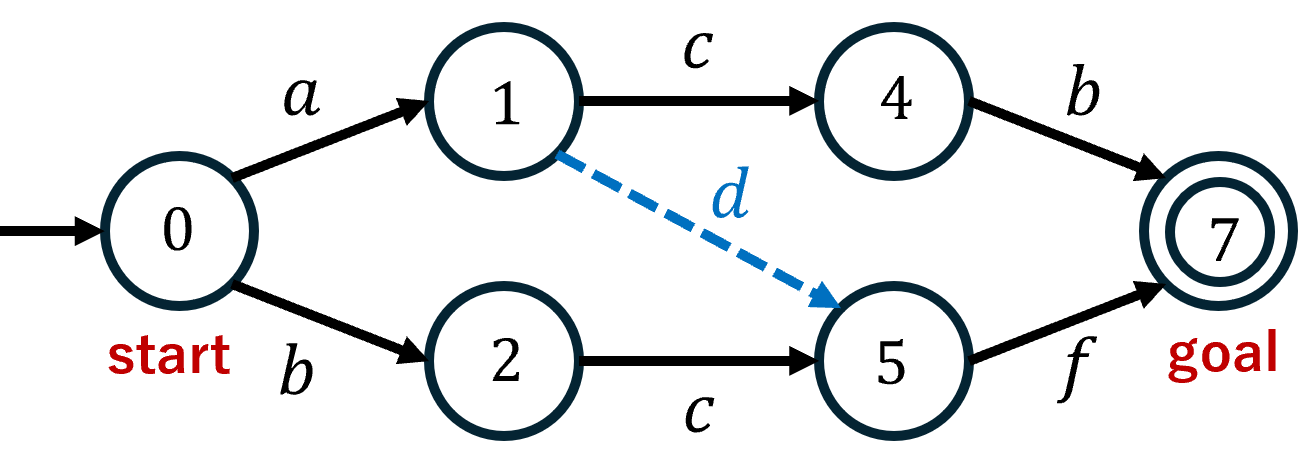}
  \caption{Robot navigation: specification $K_1(\subseteq L_m(G_1))$.  \label{fig:spec for real plant}}
\end{figure}

\begin{example} \label{ex1:standard supervisory control}
    For illustration, we provide a running example of robot navigation. Consider a robot that moves from a start point to a goal point. There exist some paths
    into a dangerous zone along the route and the robot must avoid them while heading for the goal. Also, the robot may move uncontrollably at some location due to possible disturbance from the environment. An automaton modeling this described scenario is displayed in Fig. \ref{fig:real_plant}, and we consider this automaton (say $G_1$) as the plant to be controlled.
    In this plant, each state written in number represents a location of the environment where the robot navigates. Here state $0$ is the start point and state $7$ is the goal point. This goal state is the marker state, which is denoted by double-circle. Each event written in alphabet represents the transition of the robot. We suppose that event $d$ (dashed arrow from state 1 to 5 and from 4 to 6) represents an uncontrollable transition, and other events are all controllable: i.e. $\Sigma_c = \{a, b, c, e, f\}$ and $\Sigma_u = \{d\}$. State 6 represents a danger state, and the control specification is to avoid this danger state. This (safety) specification may be written as a sublanguage of $L_m(G_1)$ as follows:
    \begin{align*}
        K_1 = \{acb, adf, bcf\}.
    \end{align*}
    This specification $K_1$ may be represented by the automaton shown in Fig. \ref{fig:spec for real plant}. Compared with the plant in Fig. \ref{fig:real_plant}, the specification automaton removes the (uncontrollable) transition from state 4 to the danger state 6 and the (controllable) transition from state 1 to state 3.

    Since the uncontrollable event $d$ can exit $\overline{K_1}$, the specification language $K_1$ is uncontrollable. Then, consider sublanguage $K_1'=\{adf, bcf\}$ of $K_1$ (eliminating path $acb$ from $K_1$). In this case, no uncontrollable event can exit $\overline{K_1'}$, so $K_1'$ is controllable and $ \text{sup}\hspace{0.1cm}C(K_1)=K_1'$. Thus we can construct the following marking nonblocking supervisor $V_1 : L(G_1)\rightarrow Pwr(\Sigma_c)$  such that $L_m(V_1/G_1)=K_1'$:
    \begin{align*}
    V_1(s) = \left\{
        \begin{array}{ll}
         \{c, e\} &\text{if } s \in \{a\}, \\
         \emptyset & \text{if } s \in L(G_1) \setminus \{a\}.
        \end{array}
    \right.
    \end{align*}
    This supervisor disables $c$ and $e$ at state 1 in Fig. 
\ref{fig:real_plant}.
\end{example}

In the above example, the supervisor $V_1$ is designed based on the assumption that the plant model $G_1$ is known. Now we pose this question: if $G_1$ is unknown (e.g. the robot navigates in an unknown environment), under what conditions can we still design a marking nonblocking supervisor? This question motivates us to study a data-driven approach to marking nonblocking supervisory control.

\section{Data-Driven Marking Supervisory Control and Marking Data-Informativity}
\label{chap:Data-Driven Marking Supervisory Control and Marking Data-Informativity}

Section 3 formulates the data-driven marking supervisory control problem in Section \ref{sec:prob formulate}. In Section \ref{sec:Marking Data-Informativity and Its Criterion}, we introduce the concept of marking data-informativity and establish a necessary and sufficient condition to characterize it. Based on this condition, Section \ref{sec:Verification of marking data-informativity} presents an algorithm for the verification
of marking data-informativity.

\subsection{Problem Formulation of Data-Driven Marking Supervisory Control} \label{sec:prob formulate}
Suppose that we have a plant whose automaton model $G$ is
unknown except for the event set $\Sigma(=\Sigma_c\cup\Sigma_u)$. Even under this circumstance, there are often situations where strings generated by the plant may be observed, i.e. a certain amount of output sequence data is available. Also, in the observation it is often possible to identify the strings that reach a  marker state. In addition, from prior knowledge of the event set $\Sigma$, it is often the case that there are certain output sequences that are obviously not generatable by the plant. For example, suppose we have a set of events: turn on a machine, press a switch on the machine, and the machine produces an output. Then, without knowing internal working mechanism of the machine, it is obvious that the machine cannot output anything before its power is turned on. In view of this, we assume that we can obtain three types of finite data sets $(D, D_m, D^-)$, where $D\subseteq\Sigma^*$ is the observed 
behavior from the plant, $D_m\subseteq \overline{D}$ is a subset of observed marked behavior, and $D^- \subseteq\Sigma^\ast$ is prior 
knowledge of impossible behavior of the plant. Since each string in $D$ and $D_m$ is observed from $G$, $\overline{D}$ is a subset of the closed behavior of $G$ and $D_m$ is a subset of the marked behavior of $G$: i.e. $\overline{D}\subseteq L(G),\hspace{0.1cm}D_m\subseteq L_m(G)$.  On the contrary since each string in $D^-$ is known to be impossible to be generated by $G$, $D^-$ and the closed behavior of $L(G)$ do not have any common elements: i.e. $D^-\cap L(G)=\emptyset$. As a result, $\overline{D}\cap D^- = \emptyset$. Given a control specification $E\subseteq\Sigma^*$, our goal is to design a
marking nonblocking supervisor (whenever it exists) to enforce $E$ for the unknown plant based on the data sets $(D, D_m, D^-)$.

\begin{example} \label{ex2:set up}
    Consider again the robot navigation example in Example \ref{ex1:standard supervisory control} and the plant model $G_1$ in Fig. \ref{fig:real_plant}. Now we suppose that $G_1$ is unknown except for the event set $\Sigma = \{a, b, c, d, e, f\}$, and certain observations and prior knowledge of the plant behavior are available. 
    For example, we have observed a string $bcf$ from the plant, i.e. the robot moves from the initial state to the goal state following the path $bcf$. Hence we have $D=\{bcf\}$ and $D_m=\{bcf\}$. 
    In addition, we have prior knowledge that the plant cannot generate string $e$, namely the robot can never start its navigation from location 1 and makes a first move to location 3. Thus $D^-=\{e\}$. For this triple $(D, D_m, D^-)$, there may exist infinitely many automata that can generate $D$ and mark $D_m$, while not generating $D^-$. In other words, our unknown plant $G_1$ cannot be uniquely identified based on the triple $(D, D_m, D^-)$. 
    For example, $G_2$ in Fig. \ref{fig:automaton_G2} and $G_3$ in Fig. \ref{fig:automaton_G3} cannot be distinguished from the true plant $G_1$. In order to design a supervisor for the true plant based only on $(D, D_m, D^-)$, we must construct a supervisor that is valid for all such possible plants.
    Intuitively, more observations and prior knowledge can help reduce the number of plant models that cannot be distinguished from the true one. Say if we observe an additional string 
    $adf$ (so that $D=\{adf, bcf\}$, $D_m=\{adf, bcf\}$), then $G_3$ can be ruled out from the candidate while $G_2$ is still possible.
\end{example}

\begin{figure}[H]
  \centering
  \includegraphics[width=80mm]
  {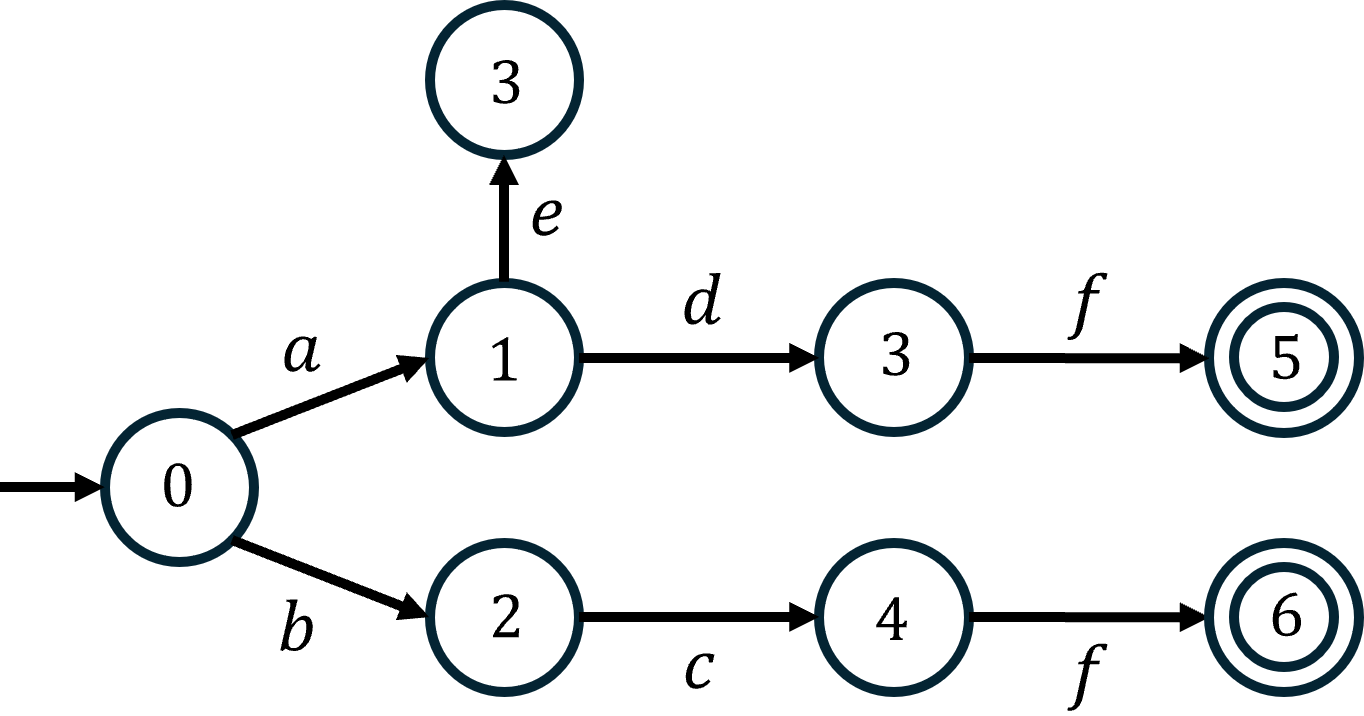}
  \caption{$G_2$ in Example \ref{ex2:set up}.  
  \label{fig:automaton_G2}}
\end{figure}

\begin{figure}[H]
  \centering
  \includegraphics[width=90mm]
  {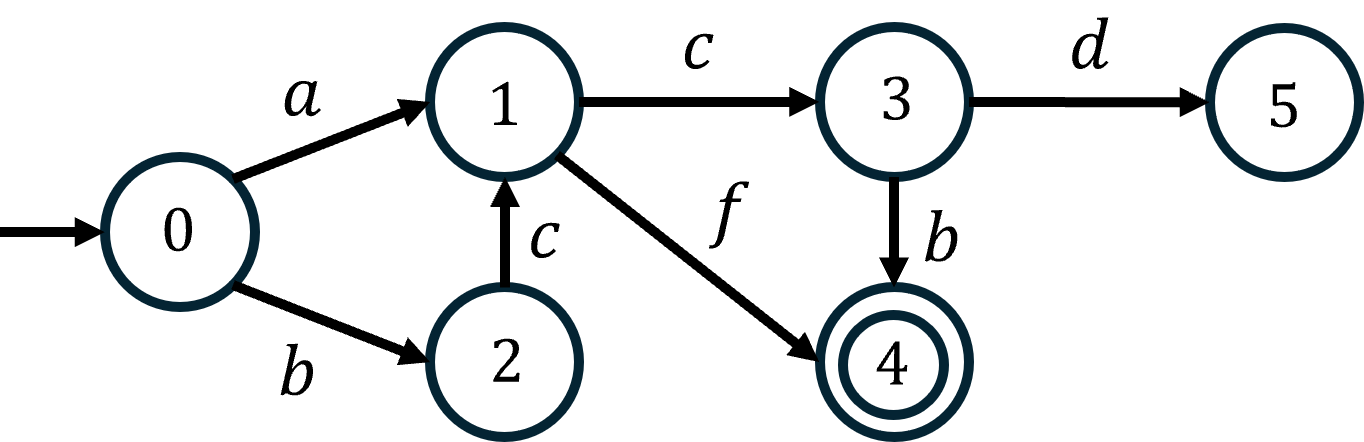}
  \caption{$G_3$ in Example \ref{ex2:set up}. 
  \label{fig:automaton_G3}}
\end{figure}

As in the example above, there are generally multiple
possible plants compatible with the given data triple $(D, D_m, D^-)$. This is defined below as a consistency property.

\begin{definition}(consistency).\label{def:consistency}
    Suppose that an event set $\Sigma$ is given. Then, for finite sets $D, D_m, D^-\subseteq\Sigma^*$ satisfying $D_m\subseteq\overline{D}$ and $\overline{D}\cap D^- = \emptyset$, an automaton $G =(Q,\Sigma,\delta,q_0,Q_m)$ is said to be consistent with $(D, D_m, D^-)$ if $\overline{D}\subseteq L(G)$, $D_m\subseteq L_m(G)$, and $D^-\cap L(G)=\emptyset$.
\end{definition}

In words, an automaton $G$ is consistent with a data triple $(D, D_m, D^-)$ if and only if all strings in $\overline{D}$ can be generated by $G$ and all strings in $D_m$ can be marked by $G$, whereas no strings $D^-$ can be generated by $G$. Thus in Example \ref{ex2:set up}, $G_1$, $G_2$ and $G_3$ are consistent with $(D, D_m, D^-)$ where $D=\{bcf\}$, $D_m=\{bcf\}$ and $D^-=\{e\}$. If we observe an additional string $adf$ ($D=\{adf, bcf\}$, $D_m=\{adf, bcf\}$ and $D^-=\{e\}$), $G_1$ and $G_2$ are
still consistent with $(D, D_m, D^-)$, but $G_3$ becomes not consistent. It is easy to verify that if we observe some additional strings or we have more knowledge about the strings that the plant cannot generate, the number of consistent models decreases.

\begin{remark} \label{re:relationship of regular language and automaton}
 It is well known \cite{hopcroft2007introduction} that for every regular language $L_m\subseteq\Sigma^*$, there exists an automaton $G$ such that $L_m(G) = L_m$. Thus for any finite sets $D, D_m, D^- \subseteq\Sigma^*$(which are regular) satisfying $D_m\subseteq\overline{D}$ and $\overline{D}\cap D^- = \emptyset$, one can always find an automaton $G$ such that $L(G)= \overline{D}$, $L_m(G)=D_m$ and hence $D^-\cap L(G)=\emptyset$. This means that there exists at least one automaton $G$ consistent with $(D, D_m, D^-)$.
\end{remark}

Before we proceed, we summarize some basic properties of consistent plants.

\begin{proposition}\label{prop:consisntent}
There exists a consistent plant with \((D, D_m, D^-)\) if and only if $D_m\subseteq \overline{D}$ is regular and \(\overline{D} \cap D^- = \emptyset\). If \(G\) is consistent with \((D, D_m, D^-)\), then \(G\) is also consistent with \((\overline{D}, D_m, D^-)\). Moreover, it holds that

\begin{align} \label{eq:consisntent}
  \bigcap \{L(G) \mid G \text{ is consistent with } (D, D_m, D^-)\} = \overline{D}.
\end{align}
\end{proposition}

 The proof of Proposition \ref{prop:consisntent} follows immediately from Definition \ref{def:consistency} and Remark \ref{re:relationship of regular language and automaton}.

The last assertion (\ref{eq:consisntent}) of Proposition \ref{prop:consisntent} motivates us to investigate the supervisory control over \( \overline{D} \). Now, we formulate the data-driven marking supervisory control problem. We denote a control specification based on \( D_m \) by:

\begin{align}
K_{D_m} := D_m \cap E, \quad E \subseteq \Sigma^* \text{ (regular language)}.\label{eq:spec besed on Dm}
\end{align}

\begin{definition} \label{def:marked language of VD/G}
Given data sets $(D, D_m, D^-)$, let $G$ be a plant consistent with $(D, D_m, D^-)$ and \( V_D : \overline{D} \to Pwr (\Sigma_c) \) be a marking noblocking supervisor for $\overline{D}$. We define the marked language of $V_D/G$ based on $D_m$ as 
\begin{align}\label{eq:marked language of VD/G}
  L_m(V_D/G) = L(V_D/G)\cap D_m.
\end{align}
\end{definition}

\begin{problem}\label{prob:data-driven marking supervisory control problem}
Suppose that we are given an event set \( \Sigma = \Sigma_c \cup \Sigma_u \), a control specification \( E \subseteq \Sigma^* \), and finite data sets \( D, D_m, D^- \subseteq \Sigma^* \) such that \( K_{D_m} \) in (\ref{eq:spec besed on Dm}) is nonempty, $D_m\subseteq\overline{D}$, and $ \overline{D} \cap D^- = \emptyset$. Construct (if possible) a marking nonblocking supervisor \( V_D : \overline{D} \to Pwr (\Sigma_c) \) such that \( L_m(V_D/G) = K_{D_m} \) for every plant \( G \) consistent with \((D, D_m, D^-)\).
\end{problem}

Since the true plant is consistent with \((D, D_m, D^-)\), the supervisor satisfying the required conditions in Problem \ref{prob:data-driven marking supervisory control problem} is valid for the true plant. If the true plant \( G \) was known, we would construct a marking nonblocking supervisor to enforce $K = L_m(G) \cap E$
(whenever \( K \) is controllable). In our data-driven setup, \( \overline{D} \) and $D_m$ each represent the accessible subsets of \( L(G) \) and $L_m(G)$ based on our observation of the plant; hence constructing a supervisor to enforce \( K_{D_m} = D_m \cap E \) is the most that can be done based on the data available. 
If the behavior represented by \( K_{D_m} \) is too small/restrictive, one can consider enlarging $\overline{D}$ and \( D_m \) by observing more behaviors of the plant. The closer $\overline{D}$ and \( D_m \) each approximate \( L(G) \) and $L_m(G)$, the closer the data-driven enforceable behavior \( K_{D_m} \) approximates the original model-based behavior \( L_m(G) \cap E \).

\subsection{Marking Data-Informativity and Its Criterion}
\label{sec:Marking Data-Informativity and Its Criterion}
As mentioned in Section \ref{sec:prob formulate}, the existence of a marking nonblocking supervisor \( V_D \) such that \( L_m(V_D/G) = K_{D_m} \) for all plants \( G \) consistent with \( (D, D_m, D^-) \) is equivalent to the controllability of the specification language \( K_{D_m} (\neq \emptyset) \). This implies that whether the available data has sufficient information can be characterized in terms of controllability.

\begin{definition}\label{def:informativity}
(Marking data-informativity). We say that \( (D, D_m, D^-) \) is {\em marking informative} for a given control specification \( E \) if there exists a marking nonblocking supervisor satisfying the required condition in Problem \ref{prob:data-driven marking supervisory control problem}, or equivalently if \( K_{D_m} \) in (\ref{eq:spec besed on Dm}) is nonempty and controllable with respect to all plants \( G \) consistent with \( (D, D_m, D^-) \).
\end{definition}

Recall that, for a known plant, if an uncontrollable event \( \sigma \in \Sigma_u \) can happen after \( s \in K_{D_m} \) in the plant, then \( s\sigma \) needs to remain in \( \overline{K_{D_m}} \) for the controllability of \( K_{D_m} \) with respect to the plant; see (\ref{eq:controllability}). On the contrary, for the data-driven case (without knowledge of the plant), we need to assume any uncontrollable event \( \sigma \in \Sigma_u \) can happen after \( s \in \overline{K_{D_m}} \) unless \( s\sigma \in D^- \) (known to be impossible). This observation leads to the following:

\begin{theorem} \label{theorem:ctiterion for informativity}
    (Criterion for marking informativity). Suppose that an event set \( \Sigma = \Sigma_c \cup \Sigma_u \) and a control specification \( E \subseteq \Sigma^* \) are given. \( (D, D_m, D^-) \) is marking informative for \( E \) if and only if
    \begin{align}
    (\forall s \in \overline{K_{D_m}}, \forall \sigma \in \Sigma_u) \quad s\sigma \in \overline{K_{D_m}} \cup D^- \label{eq:informativity}
    \end{align}
    where \( K_{D_m} \) is in (\ref{eq:spec besed on Dm}).
\end{theorem}

\textbf{Proof:} (If) Suppose that (\ref{eq:informativity}) holds. Then it is easy to verify that
\begin{align}
 (\forall s \in \overline{K_{D_m}}, \forall \sigma \in \Sigma_u) \quad s\sigma \in L(G) \implies s\sigma \in \overline{K_{D_m}}\label{eq:if of poof}
 \end{align}
 holds for every plant \( G \) consistent with \( (D, D_m, D^-) \). Thus \( (D, D_m, D^-) \) is marking informative for \( E \).

 (Only if) Suppose that \( (D, D_m, D^-) \) is marking informative for \( E \). This means by Definition \ref{def:informativity} that (\ref{eq:if of poof}) holds for every plant \( G \) consistent with \( (D, D_m, D^-) \). Consider a special plant \( G' \) such that \( L(G') = \Sigma^* \setminus D^- \). This \( G' \) is consistent with \( (D, D_m, D^-) \), and any string not in \( L(G') \) belongs to \( D^- \). Let \( s \in \overline{K_{D_m}} \) and \( \sigma \in \Sigma_u \). Consider two cases. 
 Case 1: \( s\sigma \in L(G') \). Since (\ref{eq:if of poof}) holds for \( G' \), we have \( s\sigma \in \overline{K_{D_m}} \).
 Case 2: \( s\sigma \notin L(G') \). In this case, \( s\sigma \in D^- \). Thus, (\ref{eq:informativity}) is satisfied.

Theorem \ref{theorem:ctiterion for informativity} provides a necessary and sufficient condition for marking data-informativity. Comparing (\ref{eq:informativity}) with the standard controllability condition of \( K_{D_m} \), the absence of $L(G)$ and the part of \( D^- \) mark the distinctions in the data-driven framework. 

\begin{figure}[H]
  \centering
  \includegraphics[width=100mm]
  {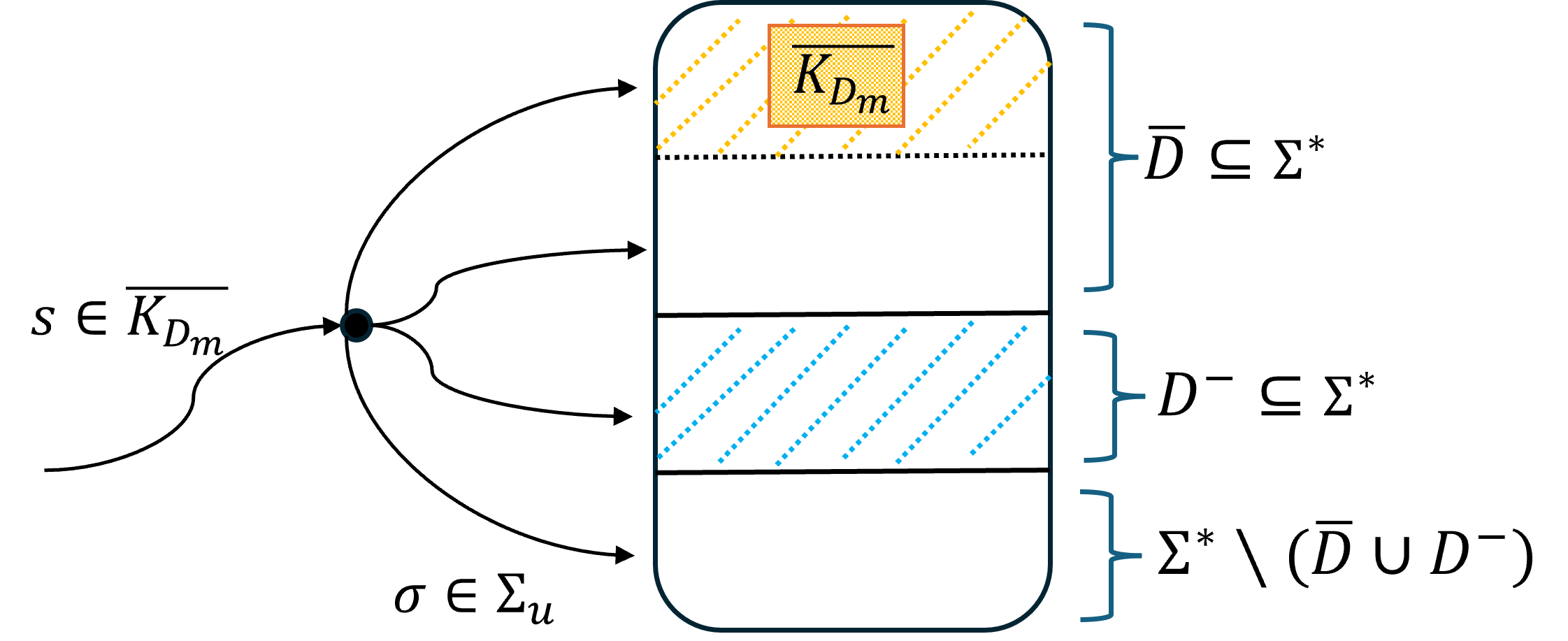}
  \caption{Illustration of transtion for $\
  s\in\overline{K_{D_m}}$ and $\sigma\in\Sigma_u$ in (\ref{eq:informativity}).
  \label{fig:understanding informativity}}
\end{figure}

To illustrate Theorem 1, write (3.4) equivalently by the following two equations:
\begin{align}
 (\forall s \in \overline{K_{D_m}}, \forall \sigma \in \Sigma_u) \quad &s\sigma \in \overline{D} \implies s\sigma \in \overline{K_{D_m}}\label{eq:controllability wrt Dbar}  \\
 \mathrm{and} \quad &s\sigma \notin \overline{D} \implies s\sigma \in D^-.\label{eq:sufficiency of D^-}
\end{align}
 
\noindent
Using only the given dataset ($D, D_m, D^-$), to determine whether $K_{D_m}$ is controllable with respect to all plants $G$ consistent with them, it is crucial to evaluate the transitions of a string $s\in \overline{K_{D_m}}$ and an event $\sigma\in\Sigma_u$, as illustrated in Figure \ref{fig:understanding informativity}. 
(3.6) represents the controllability of $K_{D_m}$ within $\overline{D}$ as a part of the unknown plant $G$ identified through observation, corresponding to the yellow region in Figure \ref{fig:understanding informativity}.
If $s\sigma\in\overline{D}\setminus\overline{K_{D_m}}$, it indicates that controllability is already violated within the observed scope, ($D, D_m, D^-$) is not marking informative. 
Unlike the model-based approach in (\ref{eq:controllability}), where any string in $\overline{K_{D_m}}$ cannot exit $\overline{K_{D_m}}$ on a continuation by an uncontrollable event, the data-driven approach requires considering the possibility of unobserved strings $s\sigma\notin\overline{D}$. 
In this case, it is necessary to ensure that such $s\sigma$ does not occur in $G$ namely $s\sigma\in D^-$. 
This requirement is precisely captured by (\ref{eq:sufficiency of D^-}), corresponding to the light blue region in Figure \ref{fig:understanding informativity}.
Finally, if $s\sigma\in\Sigma^\ast\setminus(\overline{D}\cup D^-)$, given data set lacks sufficient information and is therefore not marking informative.

 Below we remark on the key role played by $D^-$ in affecting the quality of data set.

 \begin{remark} \label{re: quantity vs quality}
 The set \(D^-\) contains strings that cannot be generated by the unknown plant (i.e., prior knowledge of impossible behavior of the plant). If we have little such prior knowledge (i.e., \(|D^-| \to 0\)), in order for (\ref{eq:informativity}) to be satisfied, \(\overline{K_{D_m}}\) must include (almost) all one-step uncontrollable continuations of strings belonging to itself. This, in turn, requires the observation data set \(D\) and $D_m$ to be rather exhaustive with respect to the occurrence of uncontrollable events, which may be challenging to obtain in practice (therefore (\ref{eq:informativity}) difficult to satisfy). Hence, in general, a larger prior knowledge \(D^-\) effectively helps relieve the requirement on obtaining observation data \(D\) and $D_m$. Note, however, that not all strings in \(D^-\) are equally useful; according to (\ref{eq:informativity}), the useful strings in \(D^-\) are precisely those of one-step uncontrollable continuation of strings in ``\(\overline{K_{D_m}}\)''. In view of the above, it is not the sheer quantity of \(D\), $D_m$ or \(D^-\), but the quality in terms of ``matchness'' between \(D\), $D_m$ and \(D^-\) specified by (\ref{eq:informativity}) that matters for marking informativity.
\end{remark}

\begin{proposition}\label{prop:data-driven supervisor}
Suppose that we are given a finite data triple $(D, D_m, D^-)$ which is marking informative for a given specification $E$. 
Then a marking nonblocking supervisor $V_D: \overline{D} \to Pwr(\Sigma_c)$ such that $L_m(V_D/G) = K_{D_m}$ is constructed as follows:
\begin{align} 
  V_{D}(s) = 
  \begin{cases} 
    \{\sigma \in \Sigma_c \mid s\sigma \notin \overline{K_{D_m}}\} & \text{if } s \in \overline{K_{D_m}}, \\
    \emptyset & \text{if } s \in \overline{D} \setminus \overline{K_{D_m}}.
  \end{cases}
  \label{eq:data-driven supervisor}
\end{align}
\end{proposition}

The above proposition can be readily derived by employing the marking nonblocking supervisor construction method used in the model-known case (as in (\ref{eq:model-based supervisor})).

\begin{example}\label{ex3:informativity}
    We illustrate the concept of marking data-informativity again using the robot navigation example. Suppose that the plant $G_1$ in Fig. \ref{fig:real_plant} is unknown except for the event set $\Sigma = \Sigma_c \cup \Sigma_u$, where $\Sigma_c = \{a, b, c, e, f\}$ and $\Sigma_u = \{d\}$. Consider two finite data triples $(D_1, D_{m1}, D^-_1)$ and $(D_2, D_{m2}, D^-_2)$, where
    \begin{align*}
      D_1\hspace{0.22cm} &= \{ae, adf, bcf\}, \\
      D_{m1} &= \{adf, bcf\}, \\
      D^-_1\hspace{0.18cm} &= \{d, bd, aed, add, bcd, adfd, bcfd\}, \\
      D_2\hspace{0.22cm}  &= \{acb, acd, bcf\}, \\
      D_{m2} &= \{acb, bcf\}, \\
      D^-_2\hspace{0.18cm} &= \{d, bd, bcd, acbd, bcfd\}.
    \end{align*}
    Let $E = \{ae, acb, adf, bcf\}$. Then the control specifications in (\ref{eq:spec besed on Dm}) are respectively
    \begin{align*}
      K_{D_{m1}} &= D_{m1} \cap E = \{adf, bcf\}, \\
      K_{D_{m2}} &= D_{m2} \cap E = \{acb, bcf\}.
    \end{align*}
    First consider $(D_1, D_{m1}, D^-_1)$. Note that $G_1$ in Fig.~\ref{fig:real_plant} and $G_2$ in Fig.~\ref{fig:automaton_G2} are consistent with $(D_1, D_{m1}, D^-_1)$. 
    From the control specification $K_{D_{m1}} = \{adf, bcf\}$, we have $\overline{K_{D_{m1}}} = \{\epsilon, a, b, ad, bc, adf, bcf\}$ and $\overline{K_{D_{m1}}}\Sigma_u = \{d, ad, bd, add, bcd, adfd, bcfd\}$. 
    Since only $ad$ belongs to $\overline{K_{D_{m1}}}$ and other strings in  $\overline{K_{D_{m1}}} \Sigma_u$ belong to $D^-_1$, the condition (\ref{eq:informativity}) holds and $(D_1, D_{m1}, D^-_1)$ is marking informative for $E$. 
    Indeed, we can confirm the controllability of  $K_{D_{m1}}$  with respect to the consistent plants $G_1$ and $G_2$ (if they were available). 
    Correspondingly a marking nonblocking supervisor $V_{D_1} : \overline{D_1}\to Pwr(\Sigma_c)$ such that $L_m(V_{D_1}/G) = K_{D_{m1}}$ is constructed for every plant $G$ consistent with $(D_1, D_{m1}, D^-_1)$ as follows:
    \begin{align}
        V_{D_1}(s) = \left\{
        \begin{array}{ll}
             \{c, e, f\} &\text{if } s = \epsilon, \\
             \{a, b, e, f\} &\text{if } s = b, \\
             \{a, b, c, e\} &\text{if } s \in \{ad,bc\}, \\
             \{a, b, c, e, f\} &\text{if } s \in \{a, adf,bcf\}, \\
             \emptyset & \text{if } s \in \overline{D_1} \setminus \overline{K_{D_{m1}}}. \nonumber
        \end{array}
    \right.
    \end{align}
    Fig. \ref{fig:Chap3_supervisor_VD1} shows the constructed supervisor $V_{D_1}$.
    \begin{figure}[H]
        \centering
        \includegraphics[width=90mm]{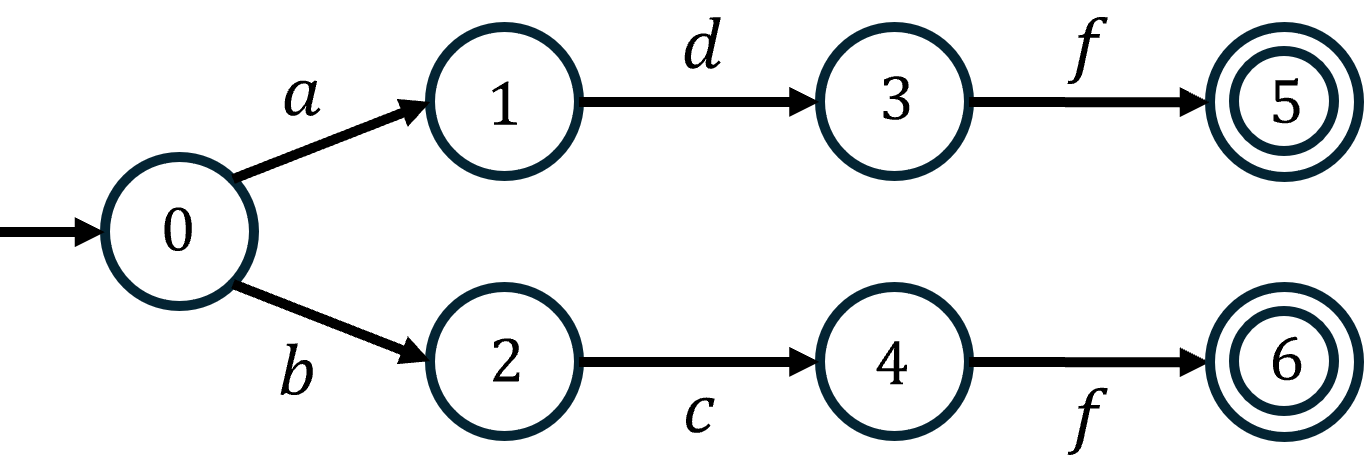}
        \caption{Supervisor $V_{D_1}$ in Example \ref{ex3:informativity}.  \label{fig:Chap3_supervisor_VD1}}
    \end{figure}

    \newpage
    \noindent
    Next, we consider \((D_2, D_{m2}, D^-_2)\). Note that \(G_1\) in Fig. \ref{fig:real_plant} and $G_3$ in Fig.
    \ref{fig:automaton_G3} are consistent with $(D_2, D_{m2}, D^-_2)$. From the control specification $K_{D_{m2}} = \{acb, bcf\}$, we have $\overline{K_{D_{m2}}} = \{\epsilon, a, b, ac, bc, acb, bcf\}$ and $\overline{K_{D_{m2}}} \Sigma_u = \{d, ad, bd, acd, bcd, acbd, bcfd\}$.
    Since $ad, acd \notin \overline{K_{D_{m2}}} \cup D^-_2$, $(D_2, D_{m2}, D^-_2)$ is not marking informative for $E$. The string $acd$ is in $\overline{D_2} \setminus \overline{K_{D_{m2}}}$, so it is outside the specification and can be generated by every consistent plant $G$. On the other hand, we do not have information about the string $ad$, so it is generatable by some consistent plant $G$ and
    not generatable by others. For example, $G_1$ in Fig. \ref{fig:real_plant} can generate $ad$ but $G_3$ in Fig. \ref{fig:automaton_G3} can not generate it.
    \end{example}

To compare the cases with and without considering marked behaviors, the data-informativity without marking and the accompanying definitions are introduced as a Remark \ref{re:data-informativity without marking} \cite{cai2022data}.

\begin{remark} \label{re:data-informativity without marking}
    When marked behaviors are not considered, the data consists only of ($D, D^-$).
    In this case, the control specification based on data is defined as 
    \begin{align}\label{eq:spec based on Dbar}
    K_D:=\overline{D}\cap E \quad (E\subseteq\Sigma^\ast).
    \end{align}
    The criterion for informativity without considering marked behavior, as an alternative to (\ref{eq:informativity}), is given by the following equation:
    \begin{align} \label{eq:informativity without marking}
        (\forall s \in \overline{K_D}, \forall \sigma \in \Sigma_u) \quad s\sigma \in \overline{K_D} \cup D^- .
    \end{align}
    The supervisor $V^{'}_D: \overline{D} \to Pwr(\Sigma_c)$ such that $L(V_D/G) = \overline{K_D}$, as an alternative to (\ref{eq:data-driven supervisor}) is constructed as follows:
    \begin{align} 
    V^{'}_{D}(s) = 
        \begin{cases} 
            \{\sigma \in \Sigma_c \mid s\sigma \notin \overline{K_{D}}\} & \text{if } s \in \overline{K_{D}}, \\
            \emptyset & \text{if } s \in \overline{D} \setminus \overline{K_{D}}.
        \end{cases}
    \label{eq:data-driven supervisor without marking}
    \end{align}
\end{remark}

By comparing the criterion for data-informativity with and without considering marked behaviors (i.e., (\ref{eq:informativity}) and (\ref{eq:informativity without marking})), the only difference lies in the replacement of the control specification $K_{D_m}$ (based on $D_m$) with $K_D$ (based on $\overline{D}$), while the structure remains the same. 
This similarity arises because, in both cases, informativity is established by ensuring two conditions: controllability within the observed scope, as represented by (\ref{eq:controllability wrt Dbar}) and (\ref{eq:sufficiency of D^-}), and the guarantee that unobserved uncontrollable events do not occur, which is ensured by $D^-$. 
While (\ref{eq:informativity}) and (\ref{eq:informativity without marking}) are very similar, the supervisors constructed under our method, which considers marked behaviors, and those constructed without considering them differ significantly in whether the constructed supervisor is guaranteed to be nonblocking. This is demonstrated in the following Example \ref{ex4:marking vs nomarking of superviosr}.

\begin{example}\label{ex4:marking vs nomarking of superviosr}
    Let us demonstrate data-informativity and constructed supervisor when  marked behaviors are not considered, using the same robot navigation example introduced in Example \ref{ex3:informativity}. Suppose the controlled unknown plant is $G_1$ in Fig. \ref{fig:real_plant} and event set $\Sigma = \Sigma_c \cup \Sigma_u$ is given. We consider a finite data set $(D_1, D^-_1)$, which is also the same as in Example \ref{ex3:informativity}.
    Let $E = \{ae, acb, adf, bcf\}$. Then the control specifications in (\ref{eq:spec based on Dbar}) is
    \begin{align*}
      K_{D_1} &= \overline{D_1} \cap E = \{ae, adf, bcf\}. \\
    \end{align*}
    Since (\ref{eq:informativity without marking}) holds with $\overline{K_{D_1}}$, ($D_1, D^-_1$) is informative.     Correspondingly a supervisor $V^{'}_{D_1}: \overline{D_1} \to Pwr(\Sigma_c)$ such that $L(V_{D_1}/G) = \overline{K_{D_1}}$, as is constructed for every plant $G$ consistent with ($D_1, D^-_1$) as follows:
    \begin{align}
    V^{'}_{D_1}(s) = \left\{
    \begin{array}{ll}
         \{c, e, f\} &\text{if } s = \epsilon, \\
         \{a, b, c, f\} &\text{if } s = a, \\
         \{a, b, e, f\} &\text{if } s = b, \\
         \{a, b, c, e\} &\text{if } s \in \{ad,bc\}, \\
         \{a, b, c, e, f\} &\text{if } s \in \{ae, adf,bcf\}, \\
         \emptyset & \text{if } s \in \overline{D_1} \setminus \overline{K_{D_1}}. \nonumber
    \end{array}
    \right.
    \end{align}
    \noindent
    Fig. \ref{fig:Chap3_supervisor_VD1 without marking} shows constructed supervisor $V^{'}_{D_1}$.
    \begin{figure}[H]
        \centering
        \includegraphics[width=90mm]{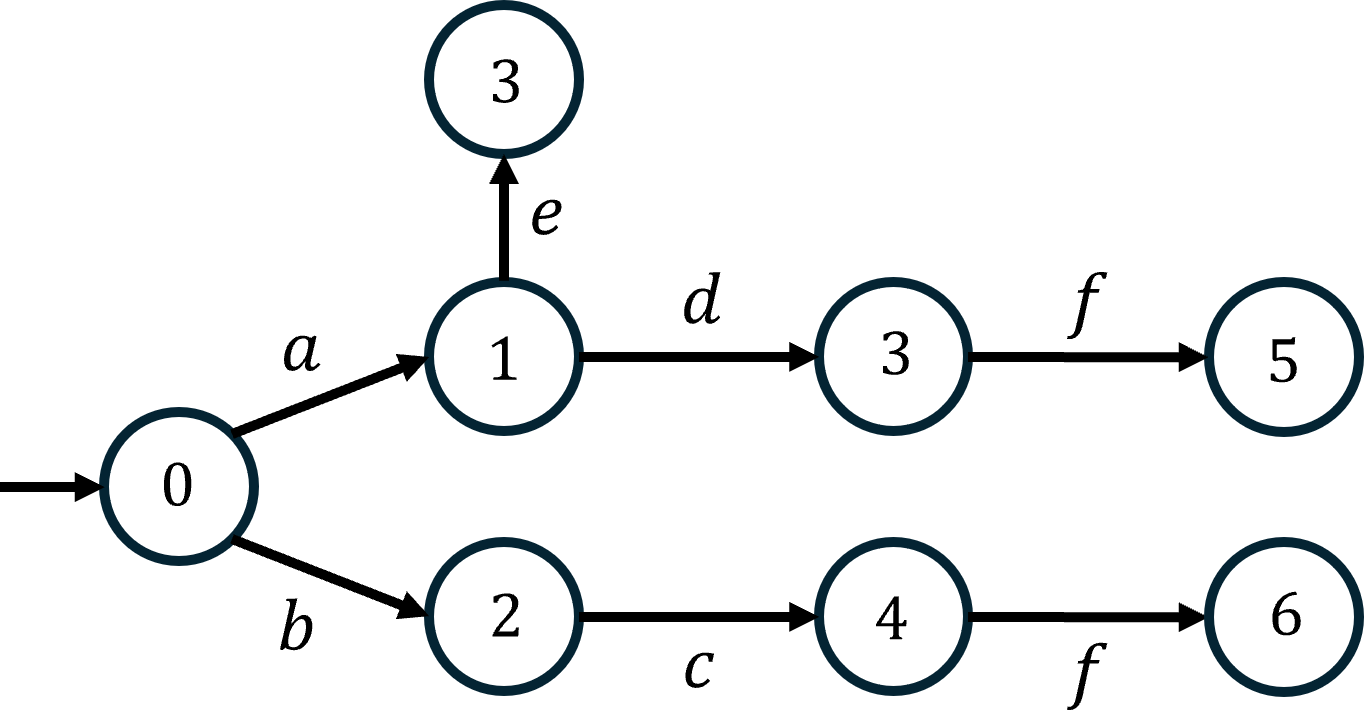}
        \caption{Supervisor $V^{'}_{D_1}$ in Example \ref{ex4:marking vs nomarking of superviosr}.  \label{fig:Chap3_supervisor_VD1 without marking}}
    \end{figure}

    By comparing Figures \ref{fig:Chap3_supervisor_VD1} and \ref{fig:Chap3_supervisor_VD1 without marking}, it can be observed that Figure \ref{fig:Chap3_supervisor_VD1 without marking} retains a path from state 1 to state 3 via event 
    $e$. This path represents a blocking scenario in the robot navigation example, where the agent fails to reach the goal. However, as shown in Figure \ref{fig:Chap3_supervisor_VD1}, such blocking paths can be eliminated by considering marked behaviors, ensuring that only paths leading to the goal remain permissible. 
    This distinction highlights the importance of considering marked behaviors in data-driven supervisory control, especially in scenarios where guaranteeing goal-oriented behavior is critical. On the other hand, supervisory control without considering marked behaviors may be more suitable for tasks where intermediate behaviors or partial achievements are sufficient, and strict goal-reaching is not a requirement. 
    From this perspective, the consideration of marked behaviors proves to be a crucial factor in supervisory control, as it ensures both goal-reaching behavior and the prevention of blocking paths, making it an essential approach in applications where achieving specific outcomes is paramount.
\end{example}

\subsection{Verification of Marking Data-Informativity}
\label{sec:Verification of marking data-informativity}

Based on the necessary and sufficient condition in Theorem \ref{theorem:ctiterion for informativity}, we next present an algorithm for checking marking data-informativity. For this purpose, we first define a {\em data-driven automaton}. We denote by \( q_s \) a state reached by a string \( s \) from the initial state of the automaton. In the data-driven automaton, a state and a string are uniquely corresponded.

\begin{definition}\label{def:DDA}
    (Data-driven automaton). 
    Suppose that the event set \( \Sigma \) and finite data sets \( D, D_m, D^-\subseteq\Sigma^* \) are given (satisfying $D_m\subseteq\overline{D} $, and $ \overline{D}\cap D^-=\emptyset$). Then a {\em data-driven automaton} is defined as follows:
    \begin{align}\label{eq:DDA}
      \hat{G}(\Sigma, D, D_m, D^-) = (\hat{Q}, \Sigma, \hat{\delta}, \hat{q}_{\epsilon}, \hat{Q}_m),
    \end{align}
    where \( \hat{Q} := \{\hat{q}_s \mid s \in \overline{D \cup D^-} \} \) is the state set, and \( \hat{\delta} := \{(\hat{q}_s, \sigma) \to \hat{q}_{s\sigma} \mid s \in \overline{D \cup D^-}, \sigma \in \Sigma, s\sigma \in \overline{D \cup D^-} \} \) is the (partial) state transition function, \( \hat{q}_{\epsilon} \) is the initial state, $\hat{Q}_m:= \{\hat{q}_s \mid s \in D_m\}$ is the marker state set. In addition, given a control specification \( K_{D_m} = D_m \cap E \) (where \( E \subseteq \Sigma^* \) is a regular language), we define $Q_K := \{\hat{\delta}(\hat{q}_{\epsilon}, s) \mid s \in \overline{K_{D_m}} \}\subseteq\hat{Q}$, $Q_- := \{\hat{\delta}(\hat{q}_{\epsilon}, s) \mid s \in D^- \}\subseteq\hat{Q}$.
\end{definition}

A data-driven automaton $\hat{G}$ is a {\em prefix tree automaton} for $\overline{D \cup D^-}$: i.e. a loop-less automaton whose closed behavior is $L(\hat{G})=\overline{D \cup D^-} $ and marked behavior is $L_m(\hat{G})=D_m$. According to the state transition function $\hat{\delta}$, for each string $s \in \overline{D \cup D^-}$, the reached state $\hat{q}_s$ is unique.
The state subset $Q_K$ contains those states reached by strings in $\overline{K_{D_m}}$. Note that since $E$ may not be a finite language in general, in order to determine $Q_K$, we first construct a (finite-state) automaton for $E$ (which is always possible since $E$ is regular), and then check if each string in the finite set $D_m$ occurs in the automaton for $E$. On the other hand, the state subset $Q_-$ contains those states reached by strings in $D^-$, meaning a transition to $Q_-$ represents an impossible behavior of the (unknown) plant. Since $\overline{D} \cap D^- = \emptyset$, we have $Q_K \cap Q_- = \emptyset$. It should be remarked that in general $Q_K \cup Q_- \neq \hat{Q}$. Also, $\hat{G}$ is not consistent with $(D, D_m, D^-)$, as $\hat{G}$ generates strings in $D^-$: i.e. $D^- \subseteq L(\hat{G})$.

\begin{example} \label{ex5:illustrate DDA}
    Here we provide examples of data-driven automatan $\hat{G}_1$ (in Fig. \ref{fig:DDA_G1}) and $\hat{G}_2$ (in Fig. \ref{fig:DDA_G2}) corresponding to the data sets $(D_1, D_{m1}, D^-_1)$ and $(D_2, D_{m2}, D^-_2)$ in Example \ref{ex3:informativity}. For clear display, we have omitted $\hat{q}_s$ in the figure, and only the subscript $s$ is written inside each state. State subsets $Q_K$ and $Q_-$ are represented in orange and blue in the figures, respectively. The states without colors correspond to strings in the observation data set $\overline{D}$, but not in the specification $E$ (thus not in $\overline{K_{D_m}}$).
\end{example}

\begin{figure}[t]
  \centering
  \includegraphics[width=80mm]
  {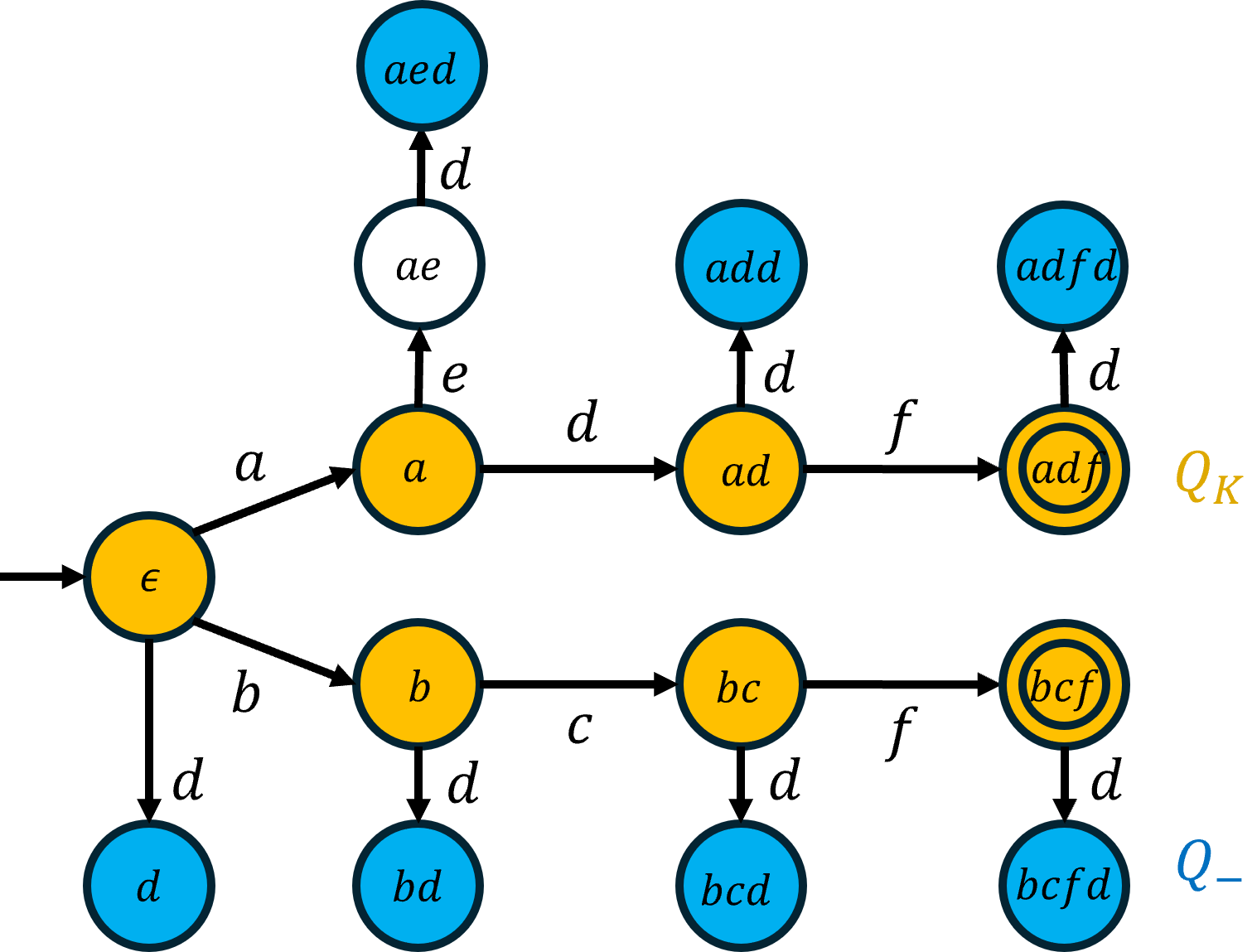}
  \caption{Data-driven automaton $\hat{G}_1$ corresponding to $(D_1, D_{m1}, D^-_1)$. \label{fig:DDA_G1}}
\end{figure}

\begin{figure}[t]
  \centering
  \includegraphics[width=80mm]
  {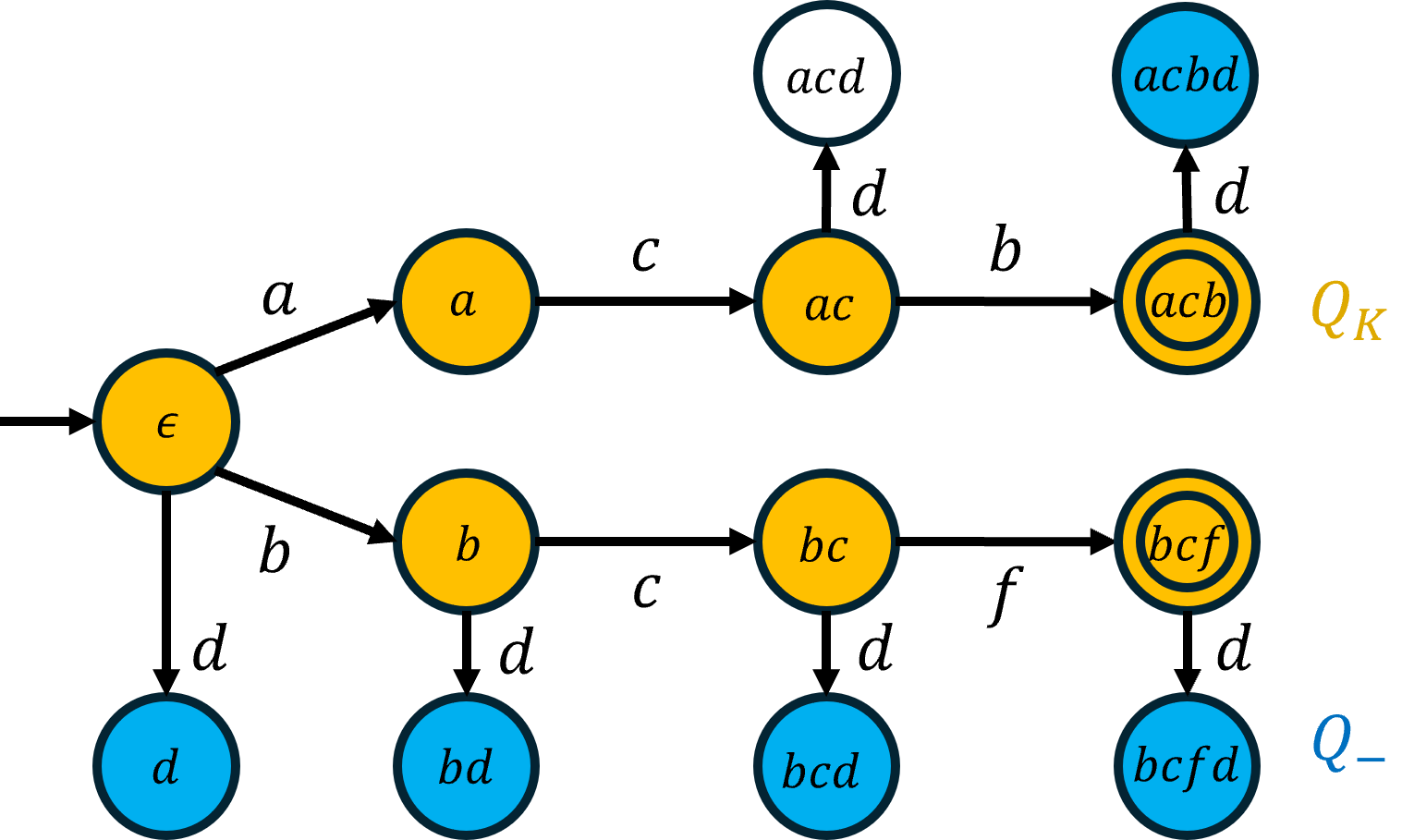}
  \caption{Data-driven automaton $\hat{G}_2$ corresponding to $(D_2, D_{m2}, D^-_2)$. \label{fig:DDA_G2}}
\end{figure}

Now we are ready to present an algorithm for verifying marking informativity based on data-driven automaton.

\clearpage

\begin{algorithm}
\caption{Checking marking informativity}
\label{alg:checking informativity}
\begin{algorithmic}[1]
\REQUIRE Event set $\Sigma = \Sigma_c \cup \Sigma_u$, finite sets $D, D_m, D^- (\subseteq \Sigma^*)$, control specification $K_{D_m} = D_m \cap E$
\ENSURE ``marking informative'' or ``not marking informative''
\STATE Construct a data-driven automaton $\hat{G}(\Sigma, D, D_m, D^-)$\\ $= (\hat{Q}, \Sigma, \hat{\delta}, \hat{q}_\epsilon, \hat{Q}_m)$ and $Q_K, Q_-$ (as in Definition \ref{def:DDA})
\FORALL{$\hat{q} \in Q_K$}
    \FORALL{$\sigma \in \Sigma_u$}
        \IF{($\hat{\delta}(\hat{q}, \sigma)!$ \AND $\hat{\delta}(\hat{q}, \sigma) \in \hat{Q} \setminus (Q_K \cup Q_-)$) \OR ($\neg\hat{\delta}(\hat{q}, \sigma)!$)}
            \RETURN ``not marking informative''
            \STATE \textbf{break}
        \ENDIF
    \ENDFOR
\ENDFOR
\RETURN ``marking informative''
\end{algorithmic}
\end{algorithm}

In Algorithm \ref{alg:checking informativity}, marking data-informativity of $(D, D_m, D^-)$ for $K_{D_m}$ is determined by examining in the data-driven automaton every uncontrollable event at each state in $Q_K$.
If an uncontrollable event $\sigma$ can occur at state $\hat{q} \in Q_K$ and the corresponding transition enters $\hat{Q} \setminus (Q_K \cup Q_-)$, then the transition is contained in $L(G)$ (for all plants $G$ consistent with $(D, D_m, D^-)$) but not contained in $\overline{K_{D_m}}$, which means that there exists a string in $\overline{K_{D_m}}$ that exits $\overline{K_{D_m}}$ by some uncontrollable event. 
Thus $\overline{K_{D_m}}$ is uncontrollable with respect to every plant $G$ consistent with $(D, D_m, D^-)$, and consequently $(D, D_m, D^-)$ is not marking informative.
If an uncontrollable event $\sigma$ cannot occur at state $\hat{q} \in Q_K$, this means that we have no data or prior knowledge about the corresponding transition, and thus we cannot determine whether the transition is generatable by the unknown true plant $G$. 
As a result, $(D, D_m, D^-)$ is not marking informative. The correctness of Algorithm \ref{alg:checking informativity} is asserted by the following proposition.

\begin{proposition} \label{prop:correctness of checking informativity}
Algorithm \ref{alg:checking informativity} returns ``marking informative'' if and only if 
$(D, D_m, D^-)$ is marking informative for $E$.
\end{proposition}

\begin{proof} 
  (If) Suppose that $(D, D_m, D^-)$ is marking informative for $E$. Then from Theorem \ref{theorem:ctiterion for informativity}, every string $s \in \overline{K_{D_m}}$ and every uncontrollable event $\sigma \in \Sigma_u$ satisfy (\ref{eq:informativity}).
  From the definition of $Q_K$ and $Q_-$, $\sigma \in \Sigma_u$ is defined at all $\hat{q}_s \in Q_K$ and $\hat{q}_{s\sigma} = \hat{\delta}(\hat{q}_s, \sigma) \in Q_K \cup Q_-$ holds. This means that Algorithm \ref{alg:checking informativity} returns ``marking informative''.

  (Only if) Suppose that Algorithm \ref{alg:checking informativity} returns ``marking informative''. Then $\hat{q}_{s\sigma}=\hat{\delta}(\hat{q}_s, \sigma) \in Q_K \cup Q_-$ holds for all $\hat{q}_s \in Q_K$ and for all $\sigma \in \Sigma_u$. If $\hat{q}_{s\sigma} \in Q_K$, then $s\sigma \in \overline{K_{D_m}}$ holds, and if $\hat{q}_{s\sigma} \in Q^-$, then $s\sigma \in D^-$ holds. Therefore, (\ref{eq:informativity}) is holds and $(D, D_m, D^-)$ is marking informative for $E$.
\end{proof}

\begin{example} \label{ex6:demonstrate Algorithm1}
  Consider the data-driven automaton $\hat{G}_1$ in Fig. \ref{fig:DDA_G1}. For state $\hat{q}_a \in Q_K$ (orange) and the (only) uncontrollable event $d$, it is satisfied that $\hat{\delta}(\hat{q}_a, d) \in Q_K$ (orange).
  For every other state $\hat{q}_s \in Q_K$ (orange) and the uncontrollable event $d$, it is satisfied that $\hat{\delta}(\hat{q}_s, d) \in Q_-$ (blue). Thus, Algorithm \ref{alg:checking informativity} returns marking informative, which corresponds to the result in Example \ref{ex3:informativity}.
  Next, consider the data-driven automaton $\hat{G}_2$ in Fig. \ref{fig:DDA_G2}. For state $\hat{q}_{ac} \in Q_K$ (orange) and the uncontrollable event $d$, we see that $\hat{\delta}(\hat{q}_{ac}, d)!$ and $\hat{\delta}(\hat{q}_{ac}, d) \notin Q_K \cup Q_-$ (since the transition enters a white state). As a result, Algorithm \ref{alg:checking informativity} returns not marking informative, which again corresponds to the result in Example \ref{ex3:informativity}. In fact, the same conclusion can be drawn based on state $\hat{q}_{a} \in Q_K$ (orange); here $\neg\hat{\delta}(\hat{q}_{a}, d)!$, so Algorithm \ref{alg:checking informativity} returns not marking informative.
\end{example}

We end this section with a note on the quantity versus the quality of the dataset  $(D, D_m, D^-)$. On one hand, enlarging the set $D, D_m$ (by making more observations) can reduce the number of models that are indistinguishable from the real plant, as well as allow more behaviors to be enforced. On the other hand, by Theorem \ref{theorem:ctiterion for informativity} (and also Algorithm \ref{alg:checking informativity}), a larger $D, D_m$ means that more strings need to be checked against the condition (\ref{eq:informativity}), and thus marking data-informativity is more difficult to hold (unless the prior knowledge data $D^-$ can also be enlarged accordingly). Hence, marking data-informativity is concerned not just with the sheer quantity of the data, but with the matching quality between the observation $D, D_m$ and the prior knowledge $D^-$ (in the sense of satisfying (\ref{eq:informativity}) as we pointed out in Remark \ref{re: quantity vs quality}). In case that $(D, D_m, D^-)$ fails to be marking informative for a specification $E$ and the prior knowledge $D^-$ cannot be enlarged, then rather than considering making more observations for $D, D_m$, one should look for a smaller subset $K \subseteq \overline {K_{D_m}}$ such that $(D, D_m, D^-)$ may be marking informative. This problem is studied in the next Section.

\section{Restricted Marking Data-Informativity}
\label{chap:Restricted marking data-informativity}
Section \ref{chap:Restricted marking data-informativity} addresses the case where the data set \((D, D_m, D^-)\) is not  marking informative for a given specification by introducing the concept of restricted marking data-informativity. This concept ensures that \((D, D_m, D^-)\) is marking informative for a smaller subset of the specification, allowing the construction of a valid supervisor to satisfy the reduced specification.

\subsection{\textit{K}-informativity}\label{sec:K-informatiity}
Given an event set $\Sigma(= \Sigma_c \cup \Sigma_u)$, finite data sets $D, D_m, D^- \subseteq \Sigma^*$, and a control specification $K_{D_m}$ in (\ref{eq:spec besed on Dm}), if $(D, D_m, D^-)$ is verified to be not marking informative for $E$, then there exists no supervisor to solve Problem \ref{prob:data-driven marking supervisory control problem}. However, it is still possible that $(D, D_m, D^-)$ is marking informative for a smaller subset $K \subseteq \overline{K_{D_m}}$. If this is the case, a valid supervisor may be constructed to enforce the smaller specification $K$ for all the plants consistent with $(D, D_m, D^-)$ (including the real plant). In this section, we formulate the notion of \textit{restricted marking data-informativity}, which aims to establish marking informativity by constraining the specification to a smaller subset. Whether or not this restricted marking informativity holds hinges on the particular subset $K$ in question, so we simply term it \textit{K}-informativity and define it as follows.

\begin{definition}(\textit{K}-informativity). \label{def:K-informativity}
Consider a specification $K_{D_m}$ in (\ref{eq:spec besed on Dm}) and let $K \subseteq K_{D_m}$. We say $(D, D_m, D^-)$ is \textit{K}-informative if there exists a marking nonblocking supervisor $V_K : \overline{D} \to Pwr(\Sigma_c)$ such that $L_m(V_K/G) = K$ for every plant $G$ consistent with $(D, D_m, D^-)$, or equivalently, if $K$ is controllable with respect to every plant $G$ consistent with $(D, D_m, D^-)$.
\end{definition}

By this definition, if $(D, D_m, D^-)$ is already marking informative for the specification $E \subseteq \Sigma^*$, then $(D, D_m, D^-)$ is $K_{D_m}$-informative ($K_{D_m} = D_m \cap E$). On the other hand, if $K = \emptyset$, since $\emptyset$ is trivially controllable, $(D, D_m, D^-)$ is always $\emptyset$-informative. However, enforcing $\emptyset$ (i.e., empty behavior) is of little practical use, so we will henceforth only consider nonempty $K$.

If $(D, D_m, D^-)$ is \textit{K}-informative for a given $K$, then the marking noblocking supervisor to realize $K$ can be constructed for every plant consistent with $(D, D_m, D^-)$ in the same way as Proposition \ref{prop:data-driven supervisor}.

For the verification of \textit{K}-informativity, a straightforward modification of Algorithm~\ref{alg:checking informativity} suffices. This is asserted below, as a corollary of Proposition~\ref{prop:correctness of checking informativity}.

\begin{corollary}\label{corollary:K-informativity}
Suppose that we are given an event set $\Sigma = \Sigma_c \cup \Sigma_u$, finite data sets $D, D_m, D^- \subseteq \Sigma^*$, a control specification $E \subseteq \Sigma^*$ and a subset $K \subseteq K_{D_m}$ with $K_{D_m}$ in (\ref{eq:spec besed on Dm}). Then Algorithm \ref{alg:checking informativity} with $Q_K$ redefined as $Q_K := \{ \hat{\delta}(\hat{q}_\epsilon, s) \mid s \in \overline{K} \}$
returns “marking informative” if and only if $(D, D_m, D^-)$ is \textit{K}-informative.
\end{corollary}

We note that in general $K$-informativity of $(D, D_m, D^-)$ does not imply $K'$-informativity for $K' \subseteq K$ (similar to the fact that a sublanguage of a controllable language need not be controllable). This is illustrated by Example \ref{ex7:K-informativity not imply K^'-informativiey}.

\begin{example} \label{ex7:K-informativity not imply K^'-informativiey}
    Consider the event set $\Sigma=\Sigma_c\cup\Sigma_u$,
    where $\Sigma_c=\{a, b\}$ and $\Sigma_u=\{d\}$, and finite data sets ($D_3, D_{m3}, D_3^-$), where 
    \begin{align*}
      D_3\hspace{0.22cm} &= \{abd\}, \\
      D_{m3} &= \{ab, abd\}, \\
      D^-_3\hspace{0.18cm} &= \{d, ad, abdd\}. 
    \end{align*}
    Let $E = \{ab, abd\}$. Then the control specification in (\ref{eq:spec besed on Dm}) is $K_{D_{m3}} = \{ab, abd\}$. The corresponding data-driven automaton \( \hat{G}_3 \) is shown in Figure \ref{fig:DDA_G3}.
    \begin{figure}[H]
      \centering
      \includegraphics[width=80mm]
      {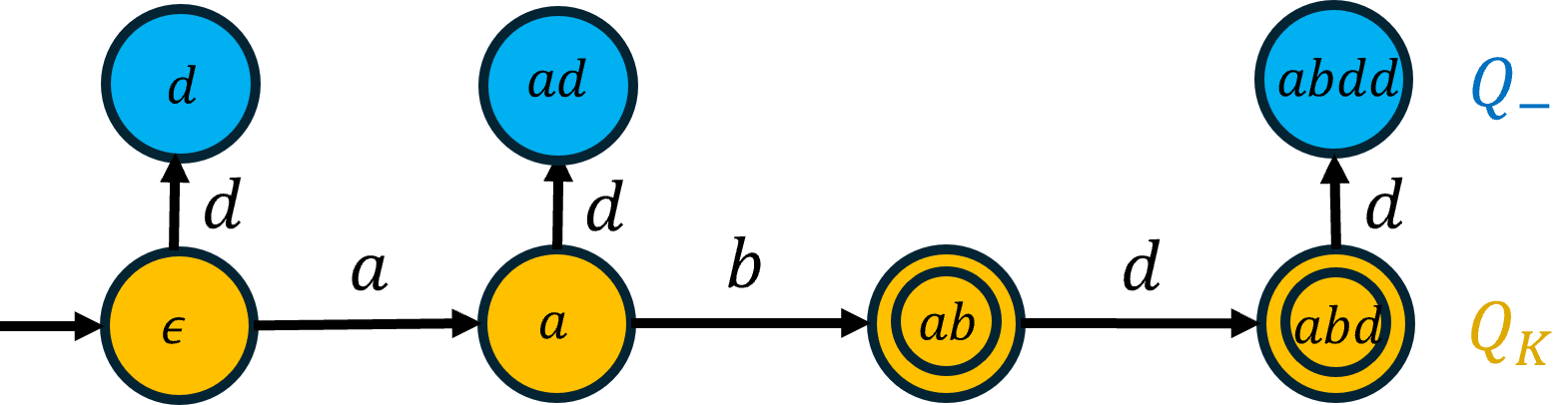}
      \caption{Data-driven automaton $\hat{G}_3$ corresponding to $(D_3, D_{m3}, D^-_3)$. \label{fig:DDA_G3}}
    \end{figure}
    By applying Algorithm \ref{alg:checking informativity}, for every state $\hat{q}_s \in Q_K$ and the uncontrollable event $d$, it is satisfied that $\hat{\delta}(\hat{q}_s, d) \in Q_K\cup Q_-$, hence ($D_3, D_{m3}, D_3^-$) is marking informative for $E$. Therefore, ($D_3, D_{m3}, D_3^-$) is $K_{D_{m3}}$-informative. However, let $K^{'}_{D_{m3}}:=\{ab\}$ (so $K^{'}_{D_{m3}}\subseteq K_{D_{m3}}$). Applying the modified Algorithm \ref{alg:checking informativity} as in Corollary \ref{corollary:K-informativity} returns ``not marking informative''. This is because $\hat{q}_{ab}\in Q_K$ but $\hat{\delta}(\hat{q}_{ab}, d)\notin Q_K \cup Q_-$. Consequently, ($D_3, D_{m3}, D_3^-$) is not $K^{'}_{D_{m3}}$-informative.

\end{example}

\subsection{Marking Informatizability}
\label{sec:Informatizability}
While in the preceding Section \ref{sec:K-informatiity} \textit{K}-informativity is defined and checked for a given subset $K \subseteq K_{D_m}$, we investigate in this section whether or not such a nonempty subset $K$ exists. This problem is formulated below.

\begin{problem}\label{prob:K-informativity}
Suppose that we are given an event set $\Sigma = \Sigma_c \cup \Sigma_u$, a control specification $E$, and finite data sets $D, D_m, D^- \subseteq \Sigma^*$ such that $K_{D_m} \neq \emptyset$, $D_m\subseteq\overline{D}$ and $\overline{D} \cap D^- = \emptyset$. Determine whether or not there exists a nonempty sublanguage $K \subseteq K_{D_m}$ such that $(D, D_m, D^-)$ is \textit{K}-informative.
\end{problem}

We term the solvability of Problem \ref{prob:K-informativity} as a property of the data $(D, D_m, D^-)$.

\begin{definition} (Marking informatizability). \label{def:informatizablity}
We say that $(D, D_m, D^-)$ is marking informatizable for a given control specification $E$ if there exists a nonempty sublanguage $K \subseteq K_{D_m}$ such that $(D, D_m, D^-)$ is \textit{K}-informative.
\end{definition}

In the case where marked behavior is not considered, a criterion for informatizability has been proposed, along with a verification algorithm based on the criterion \cite{ohtsuka2023data}. However, when marked behavior is taken into account, a similar criterion cannot be directly applied. 

Furthermore, we demonstrate the inapplicability of a similar criterion in the context of marked behavior. Specifically, we provide a remark introducing the criterion used in prior research and confirm its limitations through a concrete example. 

\begin{remark} \label{re:criteon for informatizability in prior study}
We introduce the criterion for informatizability in the case where marked behavior is not considered \cite{ohtsuka2023data}. Note that the control specification used in the following criterion is introduced in (\ref{eq:spec based on Dbar}) (i.e., based on $\overline{D}$), rather than the one introduced in (\ref{eq:spec besed on Dm}).

Suppose an event set $\Sigma = \Sigma_c \cup \Sigma_u$ and a control specification $E$ are given. Then, $(D, D^-)$ is informatizable for $E$ if and only if:
\begin{align}\label{eq:criterion for informatizability without marked behavior}
(\forall s \in \overline{K_D} \cap \Sigma_u^*, \forall \sigma \in \Sigma_u) \quad s\sigma \in \overline{K_D} \cup D^-.
\end{align}
\end{remark}

\begin{example} \label{ex8:Criterion for informatizability with marked behavior}
    Consider the event set $\Sigma=\Sigma_c\cup\Sigma_u$,
    where $\Sigma_c=\{a, b\}$ and $\Sigma_u={d}$, and finite data sets ($D_4, D_{m4}, D_4^-$), where 
    \begin{align*}
      D_4\hspace{0.22cm} &= \{a, db\}, \\
      D_{m4} &= \{a, db\}, \\
      D^-_4\hspace{0.18cm} &= \{d, ad, dd\}. 
    \end{align*}
    Let $E = \{a, db\}$. Then the control specification in (\ref{eq:spec besed on Dm}) is $K_{D_{m4}} = \{a, db\}$. The corresponding data-driven automaton \( \hat{G}_4 \) is shown in Figure \ref{fig:DDA_G4}.
    
    \begin{figure}[H]
      \centering
      \includegraphics[width=70mm]
      {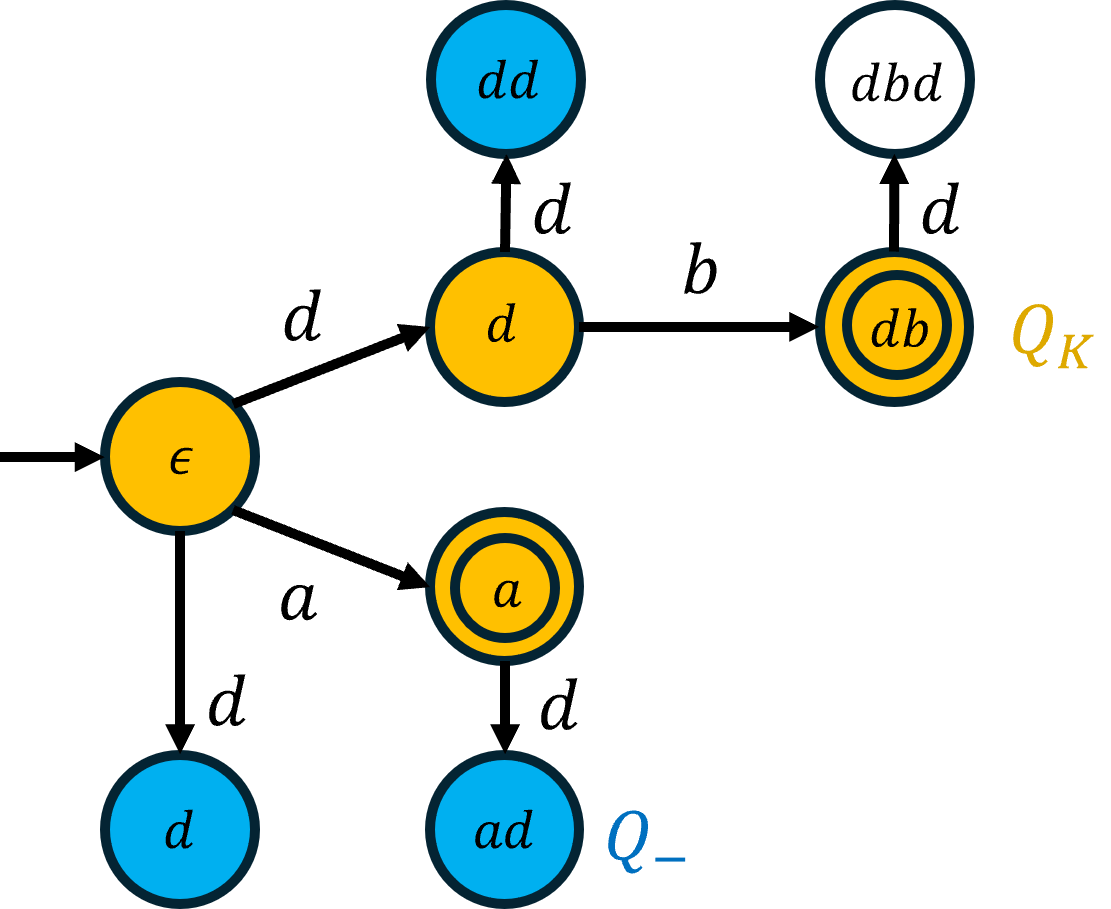}
      \caption{Data-driven automaton $\hat{G}_4$ corresponding to $(D_4, D_{m4}, D^-_4)$. \label{fig:DDA_G4}}
    \end{figure}
    
    \noindent
    By first applying the Algorithm \ref{alg:checking informativity}, it follows that for state $\hat{q}_{db} \in Q_K$ and the uncontrollable event $d$, we see that $\hat{\delta}(\hat{q}_{db}, d)!$ and $\hat{\delta}(\hat{q}_{db}, d) \notin Q_K \cup Q_-$, and therefore, ($D_4, D_{m4}, D^{-}_4$) is not marking informative for $E$. Next, we verify that (\ref{eq:criterion for informatizability without marked behavior}) holds when $K_D$ replaced with $K_{D_{m4}}$. Since $\overline{K_{D_{m4}}} = \{\epsilon, a, d, db\}$, $\overline{K_{D_{m4}}}\cap\Sigma^{\ast}_u = \{\epsilon, d\}$; and because of $\{d, dd\}\subseteq  D^-$, (\ref{eq:criterion for informatizability without marked behavior}) holds. 
    However, we conclude that ($D_4, D_{m4}, D^{-}_4$) is not marking informatizable. This is because, to prohibit the state \( \hat{q}_{db} \), which is the cause of the violation of marking informativity in $\hat{G}_4$, it is necessary to prohibit the transition $b$ on the controllable event from state $\hat{q}_{d}$ to state $\hat{q}_{db}$. Consequently, if the transition from state $\hat{q}_\epsilon$ to state \( \hat{q}_d \) remains, it will lead to blocking, and thus, this transition must also be prohibited. However, since \( d \) is an uncontrollable event, it cannot be prohibited. Therefore, it can be concluded that there is no sublanguage $K\subseteq K_{D_{m4}}$ such that ($D_4, D_{m4}, D^{-}_4$) is $K$-informative.
\end{example}

As illustrated in Example \ref{ex8:Criterion for informatizability with marked behavior}, in the proposed data-driven marking superviosory control, the constructed supervisor must be nonblocking. Consequently, the condition for the given dataset to be marking informatizable becomes stricter compared to the case where marked behavior is not considered (specifically, while condition (4.1) is necessary, it is not sufficient). To address this challenge, the next chapter proposes an algorithm to determine marking informatizability by computing the largest sublanguage $K\subseteq K_{D_m}$ such that ($D, D_m, D^-$) is $K$-informative and evaluating whether it is empty.

\section{Least Restricted Marking Data-Informativity}
\label{chap:Least restricted marking data-informativity}

If marking informatizability of the data set $(D, D_m, D^-)$ is verified to hold, then the existence of a nonempty subset $K \subseteq K_{D_m}$ is assured such that $(D, D_m, D^-)$ is \textit{K}-informative. Namely, $(D, D_m, D^-)$ satisfies restricted marking informative wrt. $K$. In this Section, we further investigate how to systematically find such a nonempty $K$. Of particular interest is to find (if possible) the largest $K_{\sup} \subseteq K_{D_m}$, so that $(D, D_m, D^-)$ is least restricted marking informative wrt. $K_{\sup}$. Then the corresponding supervisor that enforces $K_{\sup}$ is the maximally permissive one in the sense of allowing the largest set of behaviors as possible.

%\subsection{Least Restricted Marking Data-Informativiity}\label{sec:Least restricted marking data-informativiity}

We start by defining the following family of subsets of $K_{D_m}$ with respect to which the data set $(D, D_m, D^-)$ is restricted marking informative:
\begin{align} \label{eq:family of K-informative}
    I(K_{D_m}) := \{ K \subseteq K_{D_m} \mid (D, D_m, D^-) \ \text{is} \ \textit{K}\text{-informative} \}.
\end{align}

Thus if $(D, D_m, D^-)$ is marking informatizable, the family $I(K_{D_m})$ contains a nonempty member. Note also that since $K_{D_m} (= D_m \cap E)$ is finite, the number of members in $I(K_{D_m})$ is finite. 

The next result is key, which asserts that the family 

\begin{proposition}\label{prop:K-informatie is closed under the union}
    Consider the family of $I(K_{D_m})$ in (\ref{eq:family of K-informative}). If $K_1, K_2\in I(K_{D_m})$, then $K_1\cup K_2\in I(K_{D_m})$.
\end{proposition}

\begin{proof}
    According to the definition of limited marking informativity, letting $s \in \overline{K_1 \cup K_2}$ and $\sigma \in \Sigma_u$, we will show that $s\sigma \in \overline{K_1 \cup K_2} \cup D^-$. Since $s \in \overline{K_1 \cup K_2} = \overline{K}_1 \cup \overline{K}_2$, either $s \in \overline{K_1}$ or $s \in \overline{K_2}$. Consider the former case $s \in \overline{K_1}$ (the latter case $s \in \overline{K_2}$ is symmetric). Since $(D, D_m, D^-)$ is $K_1$-informative, we have
\begin{align*}
    s\sigma \in \overline{K_1} \cup D^- \subseteq (\overline{K_1} \cup \overline{K_2}) \cup D^- = \overline{K_1 \cup K_2} \cup D^-.
\end{align*}
This completes the proof. 
\end{proof}

\noindent
In view of Proposition \ref{prop:K-informatie is closed under the union}, the family $I(K_{D_m})$ contains a unique largest member $K_{\sup}$, which is the union of all the members in the family:
\begin{align} \label{eq:K_sup}
    K_{\sup} := \bigcup \{K \mid K \in I(K_{D_m})\}. 
\end{align}

With respect to this $K_{\sup} \subseteq \overline{K_{D_m}}$, the data set $(D, D_m, D^-)$ is least restricted marking informative.

The supervisor that enforces $K_{\sup}$ can be constructed for every plant consistent with $(D, D_m, D^-)$ in the same way as Proposition \ref{prop:data-driven supervisor}; due to the maximum largeness of $K_{\sup}$, this supervisor is maximally permissive.

Now that we have shown the existence and uniqueness of the largest subset $K_{\sup}$, we proceed to develop an algorithm to compute $K_{\sup}$.

\subsection{Non-Informative State}
\label{sec:Non informative state}
For computing $K_{\sup}$, we first introduce a useful concept of \textit{non-informative state}.

Given an event set $\Sigma$, a control specification $E$, and finite data sets $D, D_m, D^- \subseteq \Sigma^*$ (satisfying $D_m\subseteq\overline{D}$ and $\overline{D} \cap D^- = \emptyset$), construct the corresponding data-driven automaton $\hat{G}(\Sigma, D, D_m, D^-) = (\hat{Q}, \Sigma, \hat{\delta}, \hat{q}_\epsilon, \hat{Q}_m)$ with state subsets $Q_K$ (corresponding to $K_{D_m} = D_m \cap E$) and $Q_-$ (corresponding to $D^-$) as in Definition \ref{def:DDA}. We identify those states in $Q_K$ that violate the condition (\ref{eq:informativity}) of marking informativity. These are exactly the states for which the condition in Line 4 of Algorithm \ref{alg:checking informativity} holds. According to the condition, we state the following definition.

\begin{definition}(Non-informative state). \label{def:non-informative state}
Consider the data-driven automaton \\$\hat{G}(\Sigma, D, D_m, D^-) = (\hat{Q}, \Sigma, \hat{\delta}, \hat{q}_\epsilon, \hat{Q}_m)$ with state subsets $Q_K$, $Q_-$ (as in Definition \ref{def:DDA}). We say that $\hat{q} \in Q_K$ is a non-informative state if
\begin{align} \label{eq:non-informative state}
    (\exists \sigma \in \Sigma_u) \hspace{0.1cm}\hat{\delta}(\hat{q}, \sigma) \in \hat{Q} \setminus (Q_K \cup Q_-) \quad \text{or} \quad \neg \hat{\delta}(\hat{q}, \sigma)!
\end{align}
In addition, define the set of non-informative states as follows:
\begin{align} \label{eq:non-informative state set}
    N(Q_K) := \{ \hat{q} \in Q_K \mid \hat{q} \ \text{is a non-informative state} \}. 
\end{align}
\end{definition}

The set $N(Q_K)$ in (\ref{eq:non-informative state set}) provides an alternative characterization (to Proposition \ref{prop:correctness of checking informativity}) for marking informativity as asserted below.

\begin{proposition} \label{prop:Non-informative state set and informativity}
Consider an event set $\Sigma = \Sigma_c \cup \Sigma_u$, finite data sets $D, D_m, D^- \subseteq \Sigma^*$ (satisfying $D_m\subseteq\overline{D}$, $\overline{D} \cap D^- = \emptyset$), a control specification $E \subseteq \Sigma^*$, a subset $K \subseteq K_{D_m}$ with $K_{D_m}$ in (\ref{eq:spec besed on Dm}), a corresponding data-driven automaton $\hat{G}(\Sigma, D, D_m, D^-) = (\hat{Q}, \Sigma, \hat{\delta}, \hat{q}_\epsilon, \hat{Q}_m)$, and the non-informative state set $N(Q_K)$ in (\ref{eq:non-informative state set}). Then the following hold:
\begin{center}
    $N(Q_K) = \emptyset$ if and only if $(D, D_m, D^-)$ is marking informative for $E$.
\end{center}
\end{proposition}

\begin{proof}
  According to Algorithm \ref{alg:checking informativity}, $N(Q_K) = \emptyset$ if and only if the condition in Line 4 is never satisfied, which in turn means that Algorithm\ref{alg:checking informativity} returns “marking informative”. It then follows from Proposition \ref{prop:correctness of checking informativity} that Algorithm \ref{alg:checking informativity} returns “marking informative” if and only if $(D, D_m, D^-)$ is marking informative for $E$; hence the conclusion holds.
\end{proof}

The computation of the non-informative state set can be
adapted from Algorithm~\ref{alg:checking informativity}: instead of returning “not marking informative” immediately after identifying the first non-informative state, the new algorithm checks all states in $Q_K$ against all uncontrollable events in $\Sigma_u$, and stores all identified non-informative states. This new algorithm of computing $N(Q_K)$ is presented in Algorithm~\ref{alg:non-informative state set} below.

\begin{algorithm}
\caption{Non-informative state set}
\label{alg:non-informative state set}
\begin{algorithmic}[1]
\REQUIRE Event set $\Sigma = \Sigma_c \cup \Sigma_u$, finite sets $D, D_m, D^- (\subseteq \Sigma^*)$, control specification $K_{D_m} = D_m \cap E$
\ENSURE Non-informative state set $N(Q_K)$
\STATE Construct a data-driven automaton $\hat{G}(\Sigma, D, D_m, D^-) = (\hat{Q}, \Sigma, \hat{\delta}, \hat{q}_\epsilon, \hat{Q}_m)$ and $Q_K, Q_-$ (as in Definition \ref{def:DDA})
\STATE $N(Q_K) = \emptyset$
\FORALL{$\hat{q} \in Q_K$}
    \FORALL{$\sigma \in \Sigma_u$}
        \IF{($\hat{\delta}(\hat{q}, \sigma)!$ \AND $\hat{\delta}(\hat{q}, \sigma) \in \hat{Q} \setminus (Q_K \cup Q_-)$) \OR ($\neg\hat{\delta}(\hat{q}, \sigma)!$)}
            \STATE $N(Q_K) = N(Q_K) \cup \{\hat{q}\}$
        \ENDIF
    \ENDFOR
\ENDFOR
\RETURN $N(Q_K)$
\end{algorithmic}
\end{algorithm}

The correctness of Algorithm \ref{alg:non-informative state set} is immediate from Definition \ref{def:non-informative state}.

\subsection{Algorithm for Checking Marking Informatizability and Computing $K_{sup}$}
\label{sec:Algorithm for checking informatizability and computing K_sup}
Having introduced and identified the set $N(Q_K)$ of non-informative states in the data-driven automaton $\hat{G}$, it follows from its definition (Definition \ref{def:non-informative state}) that any string in $\overline{K_{D_m}}$ that reaches a non-informative state in $N(Q_K)$ must be excluded in order to achieve restricted marking informativity.
For achieving least restricted marking informativity, namely computing $K_{\sup} \subseteq K_{D_m}$ (as in (\ref{eq:K_sup})), such exclusion (of strings entering $N(Q_K)$) must be done in a minimally intrusive manner.

We propose to compute $K_{\sup}$ in (\ref{eq:K_sup}) based on the structure of the data-driven automaton $\hat{G}$. 
Thus intuitively, the above mentioned string exclusion amounts to ‘avoiding’ the subset $N(Q_K)$ of states in the data-driven automaton $\hat{G}$. 
In order to obtain the largest $K_{\sup} \subseteq K_{D_m}$, the ‘avoidance’ of $N(Q_K)$ should be performed as ‘close’ as possible to $N(Q_K)$ and by means of removing a controllable event. 
This strategy is similar to that of computing the supremal controllable sublanguage in the standard model-based supervisory control theory. 
In view of this, we craft a supervisory control problem based on a modified structure of the data-driven automaton with the specification of avoiding $N(Q_K)$, and solve this problem by the \textbf{supcon} function in (\ref{eq:supcon}) in order to obtain $K_{\sup}$ in (\ref{eq:K_sup}).

Specifically, we define the following relevant automaton. Given an event set $\Sigma$, a control specification $E$, and finite data sets $D, D_m, D^- \subseteq \Sigma^*$, construct the corresponding data-driven automaton $\hat{G}(\Sigma, D, D_m, D^-) = (\hat{Q}, \Sigma, \hat{\delta}, \hat{q}_\epsilon, \hat{Q}_m)$ with state subsets $Q_K$ (corresponding to $K_{D_m} = D_m \cap E$) and $Q_-$ (corresponding to $D^-$) as in Definition \ref{def:DDA}. First, construct a subautomaton $G_D$ of the data-driven automaton $\hat{G}$ by removing $Q_{r_1}:= \{\hat{q}_s \mid (\forall s\in \overline{D\cup D^-}) s \notin \overline{D}\}$ 
(i.e. the state set of subautomaton is $\{\hat{\delta}(\hat{q}_{\epsilon}, s)! \mid \forall s \in \overline{D} \}$ )
and the corresponding transitions; namely
\begin{align} \label{eq:PLANT for K_sup}
    G_D = (Q_D, \Sigma, \delta_D, \hat{q}_\epsilon, \hat{Q}_m) 
\end{align}
where $Q_D := \hat{Q} \setminus Q_{r_1}$ and $\delta_D := \hat{\delta} \setminus \{ (\hat{q}, \sigma) \to \hat{q}' \mid \hat{q} \in Q_{r_1}$ or $\hat{q}' \in Q_{r_1} \}$. Note that from Definition \ref{def:DDA} and $D_m \subseteq \overline{D}$, it follows that the marker states remain unchanged even when $Q_{r_1}$ is removed from $\hat{Q}$.
Also, we have $L(G_D) = \overline{D}, L_m(G_D)=D_m$. This automaton $G_D$ will serve as the plant in \textbf{supcon} function.

Next, let $Q_{r_2} := Q_D \setminus Q_K$ and construct a subautomaton $S_D$ of $G_D$ by removing $N(Q_K) \cup Q_{r_2}$ and the corresponding transitions; namely
\begin{align} \label{eq:SPEC for K_sup}
    S_D = (Q_S, \Sigma, \delta_S, \hat{q}_\epsilon, Q_{S,m})
\end{align}
where $Q_S = Q_D \setminus (N(Q_K) \cup Q_{r_2}) = Q_K \setminus N(Q_K)$ and $\delta_S := \delta_D \setminus \{ (\hat{q}, \sigma) \to \hat{q}' \mid \hat{q} \in N(Q_K) \cup Q_{r_2} \ \text{or} \ \hat{q}' \in N(Q_K) \cup Q_{r_2} \}$, $Q_{S, m}:= \hat{Q}_m\setminus N(Q_K)\cup Q_{r_2}$. This automaton $S_D$ will serve as the specification in \textbf{supcon} function. Note that since $Q_S = Q_K \setminus N(Q_K)$, every state in $q \in Q_S$ satisfies the negation of (\ref{eq:non-informative state}):
\begin{align} \label{eq:property of S_D}
    (\forall \sigma \in \Sigma_u) \ \hat{\delta}(q, \sigma)! \ \text{and} \ \hat{\delta}(q, \sigma) \in Q_K \cup Q_-. 
\end{align}

Let the plant be \( G_D \) in (\ref{eq:PLANT for K_sup}) and the specification be \( S_D \) in (\ref{eq:SPEC for K_sup}). Using the \textbf{supcon} function, the maximally permissive supervisor is computed, which is also an automaton (indeed in this case a subautomaton of \( S_D \)):
\begin{align}\label{eq:supcon(G_D, G_S)}
    P_D=\textbf{supcon}(G_D, S_D)=(Q_P, \Sigma, \delta_P, \hat{q}_\epsilon, Q_{P,m})  
\end{align}
where $Q_P\subseteq Q_S$, $\delta_P\subseteq \delta_S$ and 
$Q_{P,m}\subseteq Q_{S,m}$ such that
\begin{align*}
    L_m(P_D)=\text{sup} \hspace{0.08cm}C(L_m(S_D))
\end{align*}
where
\begin{align}\label{eq:C(L_m(S_D))の性質}
    C(L_m(S_D))=\{K\subseteq L_m(S_D) \mid \overline{K}\Sigma_u\cap L(G_D)\subseteq \overline{K}\}.
\end{align}.

We summarize the above procedure in the form of an algorithm below. The resulting $L_m(P_D)$ is exactly the largest $K_{\sup}$ in (\ref{eq:K_sup}) we are after. Also, if $L_m(P_D)=K_{\sup}\neq \emptyset$, then ($D, D_m, D^-$) is $K_{\sup}$-informative, so ($D, D_m, D^-$) is marking informatizable. On the other hand, if $L_m(P_D)=K_{\sup}= \emptyset$, then no nonempty subset exists for which ($D, D_m, D^-$) is restricted marking informative, so ($D, D_m, D^-$) is nor marking informatizable. Hence the nonemptyness of the following Alghrithm \ref{alg:informatizability and K_sup} serves checking marking informatizability.

\begin{algorithm}
\caption{checking marking informatizability and computing $K_{\sup}$}
\label{alg:informatizability and K_sup}
\begin{algorithmic}[1]
\REQUIRE Event set $\Sigma = \Sigma_c \cup \Sigma_u$, finite sets $D, D_m, D^- (\subseteq \Sigma^*)$, control specification $K_{D_m} = D_m \cap E$
\ENSURE (``marking informatizable'' and $L_m(P_D)$) or ``not marking informatizable''
\STATE Construct a data-driven automaton $\hat{G}(\Sigma, D, D_m, D^-)$\\ $= (\hat{Q}, \Sigma, \hat{\delta}, \hat{q}_\epsilon, \hat{Q}_m)$ and $Q_K, Q_-$ (as in Definition \ref{def:DDA})
\STATE Construct a subautomaton $G_D$ of $\hat{G}$ as in (\ref{eq:PLANT for K_sup})
\STATE Compute non-informative state set $N(Q_K)$ by Algorithm \ref{alg:non-informative state set}
\STATE Construct a subautomaton $S_D$ of $G_D$ as in (\ref{eq:SPEC for K_sup})
\STATE Compute $P_D=\textbf{supcon}(G_D, S_D)$ as in (\ref{eq:supcon(G_D, G_S)})
\IF{$L_m(P_D) \neq \emptyset$ }
\RETURN ``marking informatizable'' and $L_m(P_D)$
\ELSE
\RETURN ``not marking informatizable''
\ENDIF
\end{algorithmic}
\end{algorithm}

\begin{proposition} \label{prop:Informatizability and K_sup of algorithm}
 The following statements hold for Algorithm \ref{alg:informatizability and K_sup}:

1. The language $L_m(P_D)$ returned by Algorithm \ref{alg:informatizability and K_sup} satisfies $L_m(P_D) = K_{\text{sup}}$ in (\ref{eq:K_sup}).

2. Algorithm \ref{alg:informatizability and K_sup} returns ``marking informatizable'' if and only if $(D, D_m, D^-)$ is ``marking informatizable'' for $E$.

\end{proposition}

\begin{proof}
The proof proceeds in two parts. In Part 1, we establish that the language \(L_m(P_D)\) returned by Algorithm \ref{alg:informatizability and K_sup} satisfies \(L_m(P_D) = K_{\text{sup}}\) in (\ref{eq:K_sup}). We make the following key claim:
\[
C(L_m(S_D)) = I(K_{D_m}).
\]
\noindent
Namely, the family \(C(L_m(S_D))\) in (\ref{eq:C(L_m(S_D))の性質}) of all controllable sublanguages of \(L_m(S_D)\) is exactly the same as the family \(I(K_{D_m})\) in (\ref{eq:family of K-informative}) of all subsets of \(K_{D_m}\) for which \((D, D_m, D^-)\) is restricted marking informative. Under this claim, we derive
\[
L_m(P_D) = \sup C(L_m(S_D)) = \bigcup \{K \mid K \in C(L_m(S_D))\}
= \bigcup \{K \mid K \in I(K_{D_m})\} = K_{\text{sup}}.
\]
\noindent
The last equality is from (\ref{eq:K_sup}), and our conclusion is established.

Now we prove the claim. First, let \(K \in C(L_m(S_D))\), which means that \(K_{D_m} \subseteq L(S_D)\) and \(\overline{K}\Sigma_u \cap L(G_D) \subseteq \overline{K}\). By the definition of \(S_D\) in (\ref{eq:SPEC for K_sup}), \(Q_S \subseteq Q_K\) and thus \(L_m(S_D) \subseteq K_{D_m}\), which in turn implies that \(K \subseteq K_{D_m}
\). In order to show that \(K \in I(K_{D_m})\), we must prove that \((D, D_m, D^-)\) is \(K\)-informative.
For this, let \(s \in \overline{K}\) and \(\sigma \in \Sigma_u\). If \(s\sigma \in L(G_D)\), then it follows from \(\overline{K}\Sigma_u \cap L(G_D) \subseteq \overline{K}\) that \(s\sigma \in \overline{K}\). On the other hand, if \(s\sigma \notin L(G_D)\), we derive \(s\sigma \in D^-\). To see this, consider the state \(q_s\), which is in \(Q_S = Q_K \setminus N(Q_K)\). This means that (\ref{eq:property of S_D}) holds for \(q_s\). But \(\delta(q_s, \sigma)\) cannot be in \(Q_K\) since \(s \notin L(G_D) = \overline{D} \subseteq K_{D_m}\). Thus it is only possible that \(\delta(q_s, \sigma) \in Q_-\), which implies that \(s\sigma \in D^-\). 
To summarize, for an arbitrary \(s \in \overline{K}\) and an arbitrary \(\sigma \in \Sigma_u\), we have \(s\sigma \in \overline{K} \cup D^-\), which means that \((D, D_m, D^-)\) is \(K\)-informative. Therefore, \(K \in I(K_{D_m})\), and \(C(L_m(S_D)) \subseteq I(K_{D_m})\).

It remains to prove that \(I(K_{D_m}) \subseteq C(L_m(S_D))\). Let \(K \in I(K_{D_m})\), which means that \(K \subseteq K_{D_m}\) and \((D, D_m, D^-)\) is \(K\)-informative. Consider an arbitrary string \(s \in K \subseteq K_{D_m}\). Since \((D, D_m, D^-)\) is \(K\)-informative, the state \(q_s \in Q_K \setminus N(Q_K) = Q_S\). Hence \(s \in L_m(S_D)\), and in turn we have \(K \subseteq L_m(S_D)\). 
To show that \(K \in C(L_m(S_D))\), we need to establish that \(K\) is controllable with respect to \(G_D\). To this end, let \(s \in \overline{K}, \sigma \in \Sigma_u\), and assume that \(s\sigma \in L(G_D) = \overline{D}\). It follows again from \(K\)-informativity of \((D, D_m, D^-)\) that \(s\sigma \in \overline{K} \cup D^-\). But we have assumed that \(s\sigma \in \overline{D}\), so it is impossible that \(s\sigma \in D^-\) (by \(\overline{D} \cap D^- = \emptyset\)). Hence \(s\sigma \in \overline{K}\), which means that \(K\) is controllable with respect to \(G_D\). Therefore \(K \in C(L_m(S_D))\), and $I(K_D) \subseteq C(L(S_D))$.

In view of the above, we have established our claim that 
$I(K_{D_m}) = C(L_m(S_D))$,
and the proof of Part 1 is now completed. 

Next, we proceed to Part 2, where we prove that Algorithm \ref{alg:informatizability and K_sup} returns ``marking informatizable'' if and only if \((D, D_m, D^-)\) is ``marking informatizable'' for \(E\). By the first statement of Propositon \ref{prop:Informatizability and K_sup of algorithm}, Algorithm \ref{alg:informatizability and K_sup}
correctly computes the largest subset of $K_{D_m}$ for which the data set $(D, D_m, D^-)$ is least restricted marking informative. If Algorithm \ref{alg:informatizability and K_sup} returns ``marking informatizable'', then $L_m(P_D)=K_{\sup}\neq\emptyset$, which in turn implies $(D, D_m, D^-)$ is $K_{\sup}$-informative, so $(D, D_m, D^-)$ is marking informatizable. On the other hand, if Algorithm \ref{alg:informatizability and K_sup} returns ``not marking informatizable'', then $L_m(P_D)=K_{\sup}=\emptyset$, which in turn implies no empty subset exists for which $(D, D_m, D^-)$ is restricted marking informative, which means that $(D, D_m, D^-)$ is not marking informatizable. The proof of Part 2 is completed.
\end{proof}

We provide an illustrative example for Algorithm \ref{alg:informatizability and K_sup}. 

\begin{example} \label{ex9: checking informatizability and compute K_sup}
Consider the event set $\Sigma=\Sigma_c\cup \Sigma_u$, where $\Sigma_c=\{a, b, c, e, f\}$ and $\Sigma_u=\{d\}$, and  the finite date set $(D_5, D_{m5}, D_5^-)$, where
\begin{align*}
    D_5\hspace{0.22cm} &= \{abe, acfd, aed\},\\
    D_{m5} &= \{ab, abe, acf, aed\},\\
    D_5^-\hspace{0.18cm} &=\{d, ad, acd, aedd\}.
\end{align*}
Let $E=\{ab, acf, aed\}$. Then the control specification in (\ref{eq:spec besed on Dm}) is $K_{D_{m4}}=\{ab, acf, aed\}$. The corresponding data-driven automaton $\hat{G}_5$ is shown in Fig \ref{fig:DDA_G5}.

\begin{figure}[H]
  \centering
  \includegraphics[width=80mm]
  {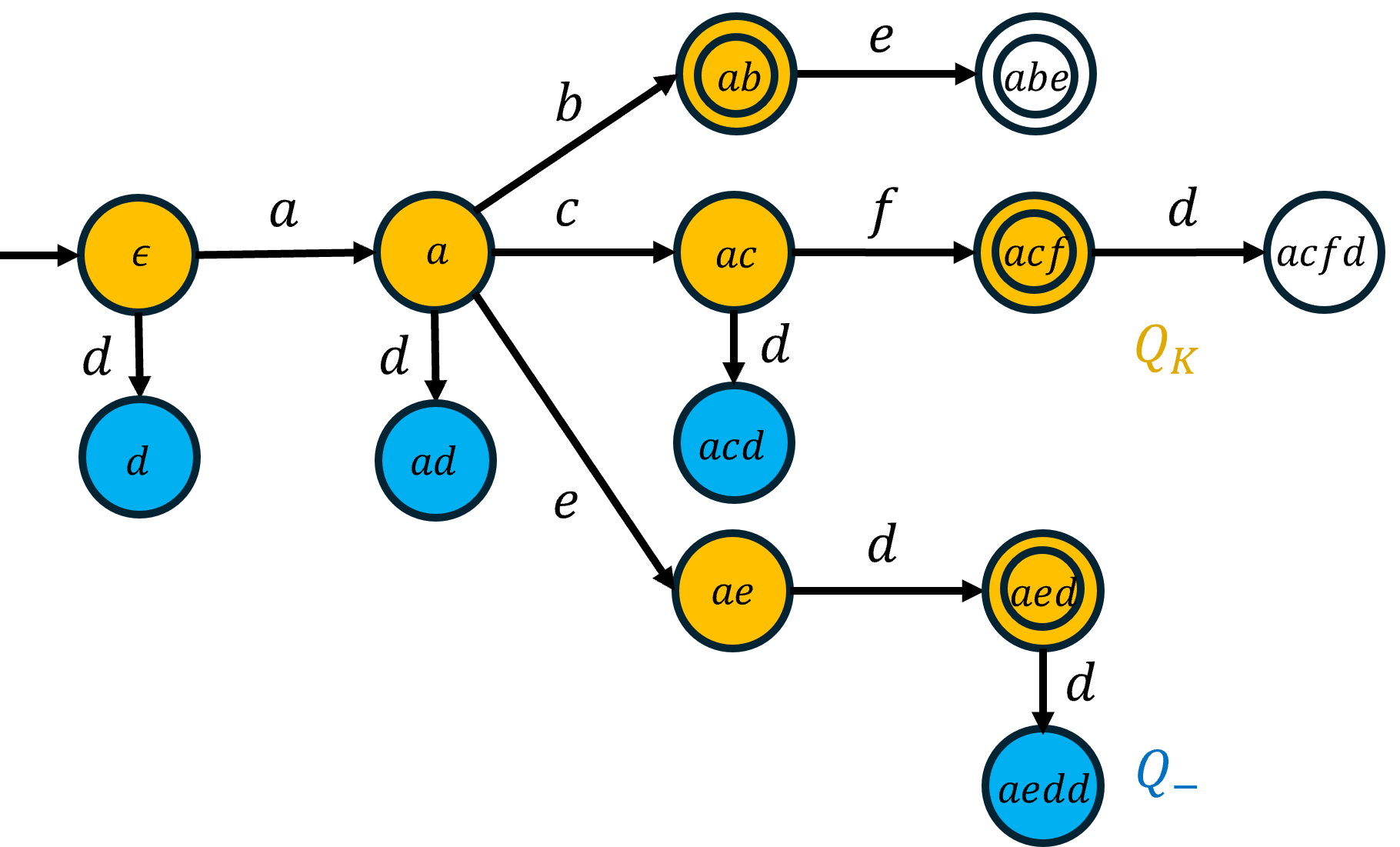}
  \caption{Data-driven automaton $\hat{G}_5$ corresponding to $(D_5, D_{m5}, D_5^-)$. \label{fig:DDA_G5}}
\end{figure}

\noindent
By applying Algorithm \ref{alg:checking informativity}, it is determined that $(D_5, D_{m5}, D_5^-)$ is ``not marking informative'' for $E$. 
However, $(D_5, D_{m5}, D_5^-)$ may be ``marking informatizable'', and there may exist a nonempty sublanguage of $K_{D_{m5}}$ for which $(D_5, D_{m5}, D_5^-)$ is limited marking informative. 
Next we apply Algorithm \ref{alg:informatizability and K_sup} to check marking informatizability and compute the largest such sublanguage.

Line 1 in Algorithm \ref{alg:informatizability and K_sup} of constructing the data-driven automaton \( \hat{G}_5 \) has been done, as displayed in Fig. \ref{fig:DDA_G5}.
Line 2 computes the subautomaton \( G_{D_5} \) by removing the states in \( Q_{r1}=Q_- \) from \( \hat{G}_5 \) and the relevant transitions (as in (\ref{eq:PLANT for K_sup})); the result is displayed in Fig. \ref{fig:PLANT G_D5}.
\begin{figure}[H]
  \centering
  \includegraphics[width=80mm]
  {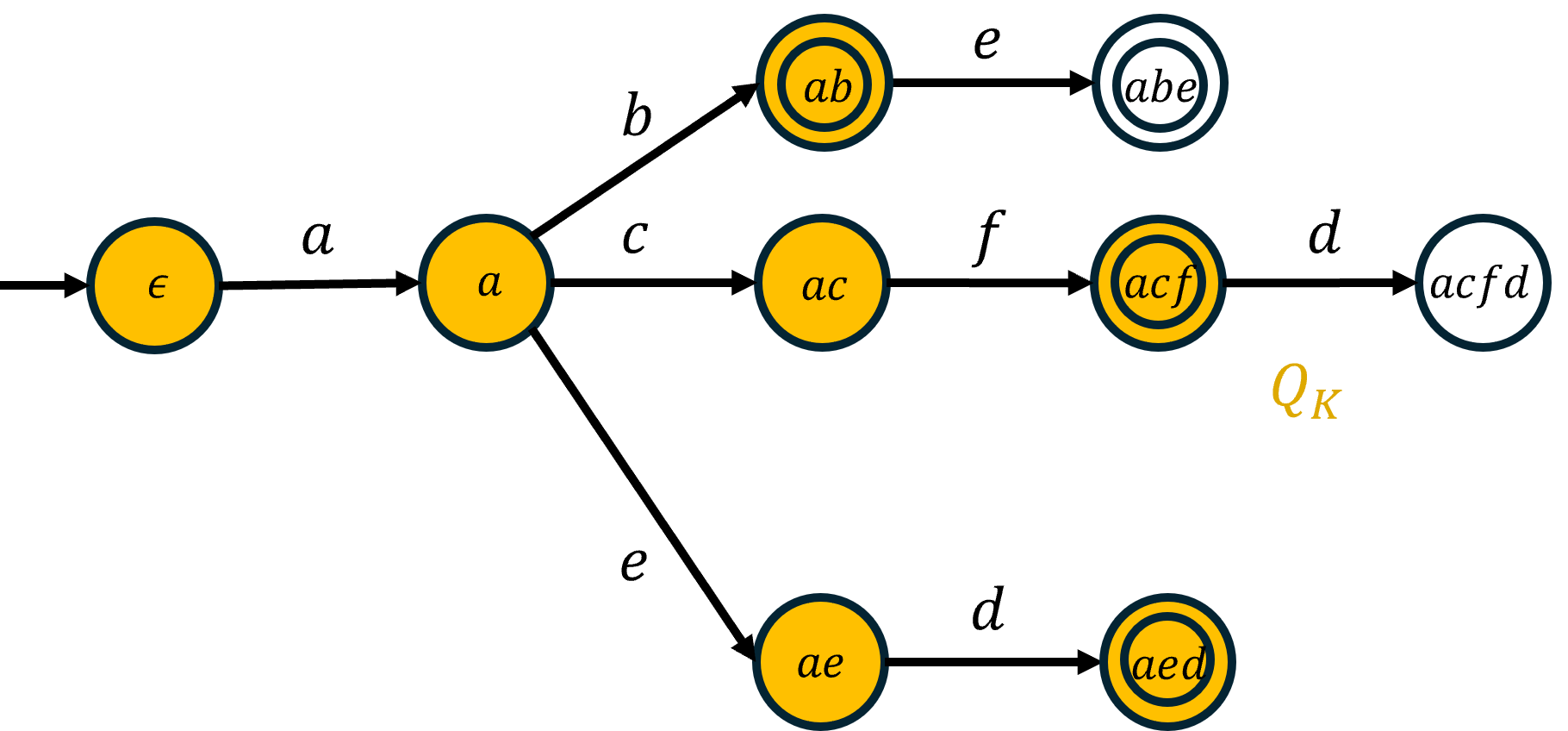}
  \caption{Subautomaton $G_{D_5}$ by Line 2 of Algorithm \ref{alg:informatizability and K_sup}. \label{fig:PLANT G_D5}}
\end{figure}

\noindent
Line 3 applies Algorithm \ref{alg:non-informative state set} to derive the set of non-informative states: \( N(Q_K) = \{ \hat{q}_{ab}, \hat{q}_{acf} \} \). The reason why these two states are non-informative states is as follows: for $\hat{q}_{ab}\in Q_K$ and the uncontrollable event $d$, $\neg\hat{\delta}(\hat{q}_{ab}, d)!$; for $\hat{q}_{acf}\in Q_K$ and the uncontrollable event $d$, $\hat{\delta}(\hat{q}_{acf}, d)!$ and $\hat{\delta}(\hat{q}_{acf}, d)\notin Q_K \cup Q_-$. 
Line 4 constructs subautomaton $S_{D_5}$ by removing the two states in $N(Q_K)$ and the two states in $Q_{r_2}$ (white color coded) from $G_{D_5}$ including the relavant transitions; the result is displayed in Fig. \ref{fig:SPEC S_D5}.
\begin{figure}[H]
  \centering
  \includegraphics[width=80mm]
  {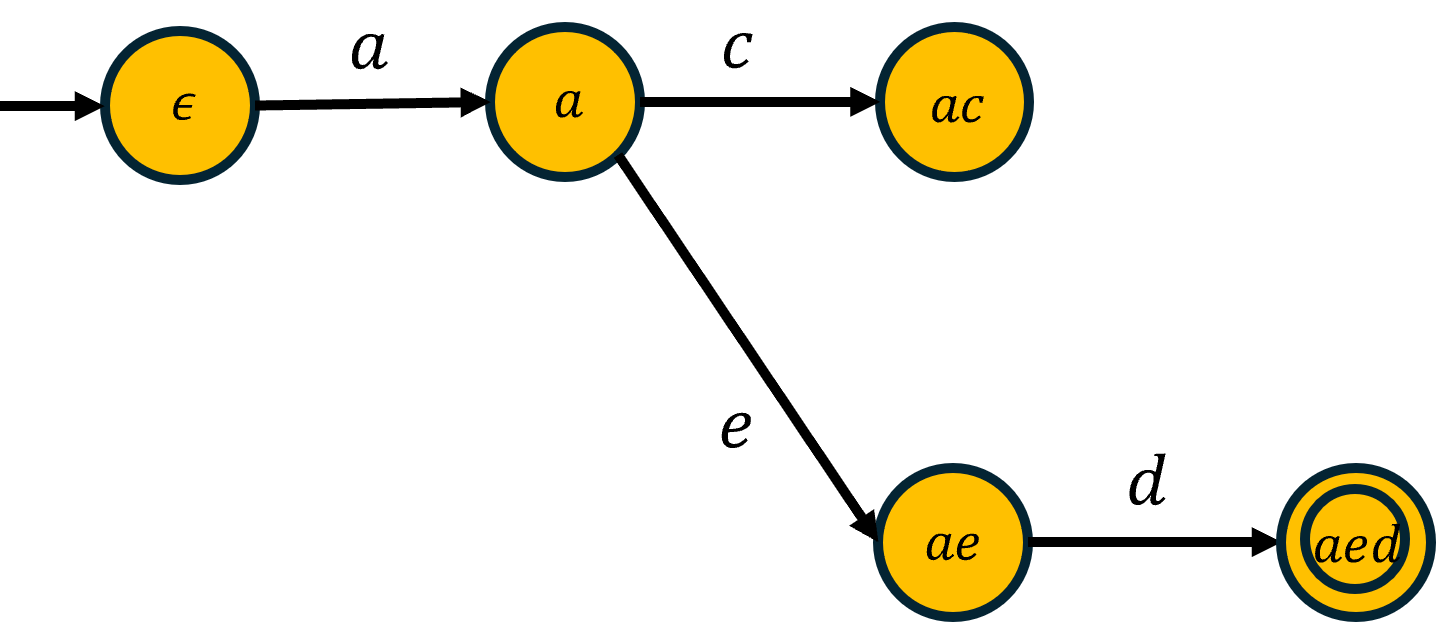}
  \caption{Subautomaton $S_{D_5}$ by Line 4 of Algorithm \ref{alg:informatizability and K_sup}. \label{fig:SPEC S_D5}}
\end{figure}

\noindent
Finally, Line 5 compute the maximally permissive supervisor $P_{D_5}=\textbf{supcon}(G_{D_5}, S_{D_5})$; the result is displayed in Fig \ref{fig:supervisor P_D5}. 
\begin{figure}[H]
  \centering
  \includegraphics[width=80mm]
  {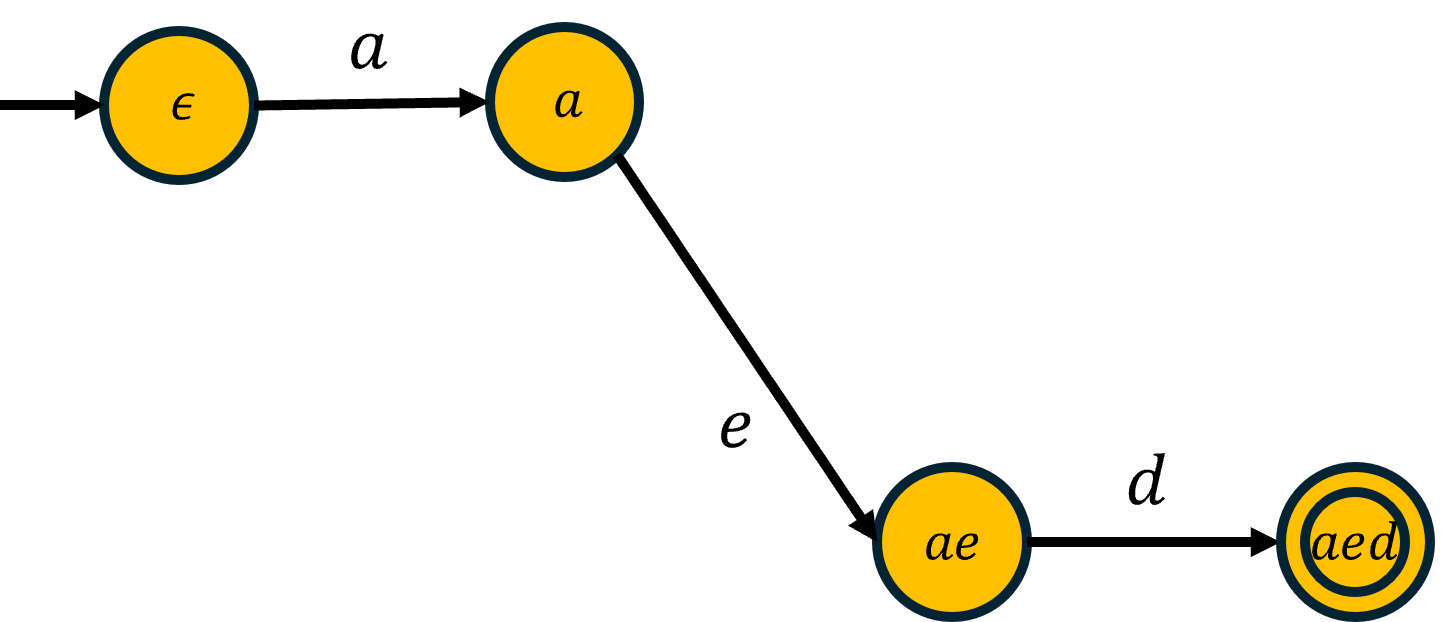}
  \caption{Supervisor $P_{D_5}$ by Line 5 of Algorithm \ref{alg:informatizability and K_sup}. \label{fig:supervisor P_D5}}
\end{figure}

\noindent
Comparing to $S_{D_5}$, $\hat{q}_{ac}$ is removed. This is because ensuring the supervisor $P_{D_5}$ is nonblocking ($\hat{q}_{ac}$ is reachable but not coreachable). Hence the final result is ($D_5, D_{m5}, D_5^-$) is marking informatizable and $L_m(P_{D_5})=\{aed\}$ which is the largest subset of $K_{D_{m5}}$ for which the data set $(D_5, D_{m5}, D_5^-)$ is least restricted marking informative.

Based on $K_{\sup}=L_m(P_{D_5})$, we can construct the corresponding supervisor $V_{\sup}:\overline{D_5}\to Pwr(\Sigma_c)$ such that $L_m(V_{\sup}/G) = K_{\sup}$ for every plant $G$ consistent with ($D_5, D_{m5}, D_5^-$) as follows:
\begin{align}
V_{\sup}(s) = \left\{
\begin{array}{ll}
     \{b, c, e, f\} &\text{if } s = \epsilon, \\
     \{a, b, c, f\} &\text{if } s = a, \\
     \{a, b, c, e, f\} &\text{if } s \in \{ae, aed\}, \\
     \emptyset & \text{if } s \in \overline{D_5} \setminus \overline{K_{\sup}}. \nonumber
\end{array}
\right.
\end{align}
\end{example}

\section{Conclusions}\label{chap:Conclusions}
In this thesis, we have studied marking data-driven supervisory control of DES. We have proposed new concepts of marking data-informativity, marking informatizability, and least restricted marking informativity, as well as developed the corresponding verification and synthesis algorithms based on a novel structure of data-driven automaton. The recipe of using these concepts/algorithms is summarized below: For a given data set \( (D, D_m, D^-) \) and a specification \( E \), first apply Algorithm~\ref{alg:checking informativity} to verify if \( (D, D_m, D^-) \) is informative for \( E \). If yes, we can build a supervisor to enforce \( K_{D_m} = D_m \cap E \). If no, we apply Algorithm~\ref{alg:informatizability and K_sup} to verify if \( (D, D_m, D^-) \) is informatizable for \( E \). If yes, we can get the largest subset of \( K_{D_m} \) for which \( (D, D_m, D^-) \) is least restricted informative. If Algorithm \ref{alg:informatizability and K_sup} returns no, then there is no supervisor that can be built to enforce any subset of \( K_{D_m} \). In this case, one may consider using Algorithm~\ref{alg:computing D^-_new} and Algorithm~\ref{alg:compute update D and Dm} to collect more data that is least restricted marking informative with respect to larger languages. 

To validate the proposed data-driven approach for marking supervisory control, we have conducted experiments on both typical examples and a scenario involving ``unknown environment exploration''. The results have shown that the complexity of the plant and the types of uncontrollable events affect the required quality of $D^-$ for achieving marking informativity. As the plant becomes more intricate or the number of uncontrollable events increases, higher-quality $D^-$ is needed to achieve that the data set is marking informative.

A promising avenue for future work is to explore how the developed data-driven approach can be further refined and extended. This includes investigating whether a supervisor can be effectively constructed using newly obtained data while leveraging the already constructed data-driven automaton. Additionally, exploring the relaxation of current assumptions, such as requiring all uncontrollable events to be considered unless they lead to $D^-$, by preempting uncontrollable events through forcing, remains an open question \cite{gu2024data}. Furthermore, extending informativeness beyond basic controllability to encompass other critical properties, such as observability, diagnosability, and opacity, is an important direction for further research \cite{hadjicostis2020estimation}.

\bibliographystyle{IEEEtran}        % Include this if you use bibtex 
\bibliography{autosam}           % and a bib file to produce the 
                                 % bibliography (preferred). The
                                 % correct style is generated by
                                 % Elsevier at the time of printing.

%\begin{thebibliography}{99}     % Otherwise use the  
                                 % thebibliography environment.
                                 % Insert the full references here.
                                 % See a recent issue of Automatica 
                                 % for the style.
%  \bibitem[Heritage, 1992]{Heritage:92}
%     (1992) {\it The American Heritage. 
%     Dictionary of the American Language.}
%     Houghton Mifflin Company.
%  \bibitem[Able, 1956]{Abl:56}
%     B.~C.~Able (1956). Nucleic acid content of macroscope. 
%     {\it Nature 2}, 7--9. 
%  \bibitem[Able {\em et al.}, 1954]{AbTaRu:54}   
%     B.~C. Able, R.~A. Tagg, and M.~Rush (1954).
%     Enzyme-catalyzed cellular transanimations.
%     In A.~F.~Round, editor, 
%     {\it Advances in Enzymology Vol. 2} (125--247). 
%     New York, Academic Press.
%  \bibitem[R.~Keohane, 1958]{Keo:58}
%     R.~Keohane (1958).
%     {\it Power and Interdependence: 
%     World Politics in Transition.}
%     Boston, Little, Brown \& Co.
%  \bibitem[Powers, 1985]{Pow:85}
%     T.~Powers (1985).
%     Is there a way out?
%     {\it Harpers, June 1985}, 35--47.

%\end{thebibliography}

%\appendix
%\section{A summary of Latin grammar}    % Each appendix must have a short title.
%\section{Some Latin vocabulary}         % Sections and subsections are supported  
                                        % in the appendices.
\end{document}